\documentclass[12pt]{article}

\usepackage[colorlinks=true,citecolor=black,linkcolor=black,urlcolor=blue]{hyperref}
\newcommand{\arxiv}[1]{\href{http://arxiv.org/abs/#1}{\texttt{arXiv:#1}}}

\usepackage{amsmath,amssymb,amsthm}
\usepackage{tikz,color,graphicx}

\newcommand*{\hatN}{\widehat{N}}

\newtheorem{theorem}{Theorem}[section]
\newtheorem{lemma}[theorem]{Lemma}

\def\nfrac#1#2{{\textstyle\frac{#1}{#2}}}
\def\dfrac#1#2{\lower0.15ex\hbox{\large$\frac{#1}{#2}$}}

\newcommand{\dvec}{\boldsymbol{d}}
\newcommand{\vecdvec}{\vec{\dvec}}
\newcommand{\Omegau}{\Omega(\dvec)}
\newcommand{\Omegad}{\Omega(\vecdvec)}
\newcommand{\Mu}{\mathcal{M}(\dvec)}
\newcommand{\Md}{\mathcal{M}(\vecdvec)}
\newcommand{\tauu}{\tau(\Mu,\varepsilon)}
\newcommand{\taud}{\tau(\Md,\varepsilon)}

\newcommand{\sdmax}{r_{\max}}
\newcommand{\sdmin}{r_{\min}}
\newcommand{\dmax}{d_{\max}}
\newcommand{\dmin}{d_{\min}}

\makeatletter

\renewcommand{\p@enumii}{}

\renewcommand{\p@enumiii}{}
\makeatother

\title{The switch Markov chain\\ for sampling irregular graphs and digraphs\footnote{An earlier version of this work, for undirected graphs only, appeared in SODA 2015~\cite{SODA}.}}

\author{Catherine Greenhill\thanks{
Research supported by the Australian Research Council, Discovery Project DP140101519.}\\
\small School of Mathematics and Statistics\\[-0.8ex]
\small UNSW Sydney\\[-0.8ex]
\small NSW 2052, Australia\\
\small \tt c.greenhill@unsw.edu.au\\
\and
Matteo Sfragara\thanks{Research supported by NWO Gravitation Grant 024.002.003--NETWORKS. }\\
\small Mathematical Institute\\[-0.8ex]
\small Leiden University\\[-0.8ex]
\small P.O. Box 9512, 2300 RA Leiden, The Netherlands\\
\small \tt m.sfragara@math.leidenuniv.nl
}

\date{12 September 2017}

\begin{document}

\maketitle

\begin{abstract}
The problem of efficiently sampling from a set of (undirected, or directed) graphs 
with a given degree sequence has many applications. One approach to this problem uses
a simple Markov chain, which we call the switch chain, to perform the sampling.  
The switch chain is known to be rapidly mixing for regular degree sequences,
both in the undirected and directed setting.

We prove that the switch chain for undirected graphs is rapidly mixing for any
degree sequence with minimum degree at least 1 and with maximum degree $\dmax$
which satisfies $3\leq \dmax\leq \nfrac{1}{3}\, \sqrt{M}$, where $M$ is the sum of the degrees. 
The mixing time bound obtained is only a factor $n$ larger than that
established in the regular case, where $n$ is the number of vertices.
Our result covers a wide range of degree sequences, including power-law
density-bounded graphs with parameter $\gamma > 5/2$ and sufficiently many edges.

For directed degree sequences such that
the switch chain is irreducible, we prove that the switch chain is 
rapidly mixing when all in-degrees and out-degrees are positive
and bounded above by 
$\nfrac{1}{4}\, \sqrt{m}$, where $m$ is the number of arcs,
and not all in-degrees and out-degrees equal 1.
The mixing time bound obtained in the directed case is an order of $m^2$
larger than that established in the regular case.

\end{abstract}

\section{Introduction}\label{s:intro}

There are several approaches to the problem of sampling from a set of graphs
(or directed graphs) with a given degree sequence. In this paper we focus
on the Markov chain approach. Here the running time of the sampling
algorithm must be (deterministically) polynomially bounded but the output need not
be exactly uniform: however, the user can specify how far from
the uniform distribution the samples may be.
Other approaches to the problem of sampling graphs (or directed graphs) 
are discussed in Section~\ref{ss:history}.

The switch chain is a natural and well-studied Markov chain for sampling from a 
set of graphs with a given degree sequence.
Each move of the switch chain selects two distinct 
edges uniformly at random and attempts to replace these edges by  
a perfect matching of the four endvertices, chosen
uniformly at random.  The proposed move is rejected if the four
endvertices are not distinct or if a multiple edge
would be formed. 
We call each such move a \emph{switch}.  
The precise definitions of the transitions for the switch chain
for undirected and directed graphs are given at the start
of Sections~\ref{s:undirected} and~\ref{s:directed}, respectively.

Ryser~\cite{ryser} used switches to study 
the structure of 0-1 matrices. Markov chains based on switches have been introduced 
by Besag and Clifford~\cite{BC89} for 0-1 matrices (bipartite graphs), Diaconis and 
Sturmfels~\cite{DS98} for contingency tables and Rao, Jana and Bandyopadhyay~\cite{RJB} 
for directed graphs.

The switch chain is aperiodic and its transition matrix is symmetric. 
It is well-known that the switch chain is irreducible for any (undirected)
degree sequence: see~\cite{petersen,taylor}. 
Irreducibility for the directed chain is not guaranteed, see Rao et al.~\cite{RJB}. 
However, Berger and M{\" u}ller-Hanneman~\cite{BMH}
and LaMar~\cite{lamar,lamar2011} gave characterisations of directed degree 
sequences for which the
switch chain is irreducible.  In particular, the switch chain is irreducible for
regular directed graphs (see for example Greenhill~\cite[Lemma 2.2]{directed}).

In order for the switch chain to be useful for sampling, it must converge quickly
to its stationary distribution.  The rate of convergence of a Markov chain
$\mathcal{M}$ is captured by its \emph{mixing time} $\tau(\mathcal{M},\varepsilon)$,
which is the minimum number of steps that the Markov chain $\mathcal{M}$ must run before
its distribution is less than $\varepsilon$ from stationarity, in total
variation distance, from a worst-case starting state.
A Markov chain with state space $\Omega$ is said to be \emph{rapidly mixing}
if its mixing time can be bounded above by some polynomial in
$\log(|\Omega|)$ and $\log(\varepsilon^{-1})$.
See Section~\ref{ss:flow} for more details.  

Cooper, Dyer and Greenhill~\cite{CDG,CDG-corrigendum}
showed that the switch chain is rapidly mixing for regular undirected graphs.
Here the degree $d=d(n)$ may depend on $n$, the number of vertices.  The mixing time
bound is given as a polynomial in $d$ and $n$. 
Earlier, Kannan, Tetali and Vempala~\cite{KTV} investigated the mixing time 
of the switch
chain for regular bipartite graphs.  
Greenhill~\cite{directed} proved that the
switch chain for regular directed graphs (that is, $d$-in, $d$-out directed graphs)
is rapidly mixing, again for any $d=d(n)$.
Mikl{\' o}s, Erd{\H o}s and Soukup~\cite{MES} proved that the switch chain is
rapidly mixing on half-regular bipartite graphs; that is, bipartite degree sequences
which are regular for vertices on one side of the bipartition, but need not be
regular for the other.

A multicommodity flow argument~\cite{sinclair}
was used in each of~\cite{CDG,directed,KTV,MES} 
to prove an upper bound on the mixing time of the switch chain. 
In each case, regularity (or half-regularity) was only required for one 
lemma, which we will call the \emph{critical lemma}.  This is a counting
lemma which is used to bound
the maximum load of the flow 
(see~\cite[Lemma 4]{CDG},~\cite[Lemma 5.6]{directed} and~\cite[Lemma 6.15]{MES}).  

In Section~\ref{s:undirected} we consider the undirected switch chain
and prove the following theorem. This extends the
rapid mixing result from~\cite{CDG} to irregular degree sequences
which are not too dense.

Given a degree sequence $\dvec = (d_1,\ldots, d_n)$, write
$\Omegau$ for the set of all (simple, undirected) graphs with vertex
set $[n] = \{ 1,2,\ldots, n\}$ and degree sequence $\dvec$.
Recall that $\dvec$ is called \emph{graphical} when $\Omegau$
is nonempty. We restrict our attention to graphical sequences.
Write $\dmin$ and $\dmax$ for the minimum and maximum degree in $\dvec$,
respectively, and let $M = \sum_{j=1}^n d_j$ be the sum of the degrees.

\begin{theorem}
Let $\dvec = (d_1,\ldots, d_n)$ be a graphical degree sequence 
such that $\dmin\geq 1$ and $3\leq \dmax\leq \frac{1}{3}\, \sqrt{M}$.
The mixing time $\tauu)$ of the switch chain $\Mu$ with
state space $\Omegau$ satisfies
\[ \tauu \leq  \dmax^{14}\, M^9 \left( \nfrac{1}{2} M\log(M) +
 \log(\varepsilon^{-1})\right).\]
\label{main}
\end{theorem}
The proof of this result given in an earlier version of this
paper~\cite{SODA} had a small gap in the proof. We have fixed
the gap here, while also improving the upper bound on $\dmax$
by a small constant factor. However, we have not made a serious
attempt to optimise the constants.

Theorem~\ref{main} covers many different degree sequences, for example:
\begin{itemize}
\item sparse
graphs with constant average degree and maximum degree a 
sufficiently small constant times $\sqrt{n}$, 
\item dense graphs with linear average degree and 
maximum degree a sufficiently small constant times $n$.
\item power-law density-bounded graphs with parameter $\gamma > 5/2$, 
when $M$ is sufficiently large. 
Such graphs were considered by Gao and Wormald~\cite{GW-power}: see in 
particular~\cite[Section 5]{GW-power},
where they prove that $\dmax = O(M^{2/5})$ for such graphs
(or in their notation, $\Delta = O(M_1^{2/5})$).
\end{itemize}
Since $M\leq \dmax n$, the mixing time bound given above is 
at most a factor of $n$ 
larger than that obtained in~\cite{CDG,CDG-corrigendum} in the regular case.

The directed case is similar, and is considered in Section~\ref{s:directed}.
To state our main result for the directed switch chain, we must
introduce some notation. 
For definitions about directed graphs not given here, see~\cite{digraphs}.

A \emph{directed degree sequence} is a sequence $\vecdvec$ of ordered pairs
of nonnegative integers
$\vecdvec = ((d_1^-,d_1^+),\ldots, (d_n^-,d_n^+))$, such that
$d_j^-$ is the in-degree and $d_j^+$ is the out-degree of vertex $j$,
for all $j\in [n]$. (We use the arrow over the symbol $\vecdvec$ so that our
notation distinguishes directed and undirected degree sequences.)
The directed degree sequence is \emph{digraphical} if there
exists a directed graph with these in-degrees and out-degrees.
Write $\Omegad$ for the set of all directed graphs
with vertex set $[n]$ such that the in-degree (respectively, out-degree)
of vertex $j$ is $d_j^-$ (respectively, $d_j^+$).
Let
\[ m = \sum_{j=1}^n d_j^- = \sum_{j=1}^n d_j^+\]
be the number of arcs in a directed graph with directed degree sequence $\vecdvec$.
We say that the directed degree sequence $\vecdvec$ is 
\emph{switch-irreducible} if the directed switch chain on $\Omegad$
is irreducible.

Finally, let $\sdmin$ and $\sdmax$ denote the minimum and
maximum \emph{semi-degree} of $\vecdvec$, defined by
\[ \sdmin = \min\{ d_1^-,\, d_1^+,\ldots, d_n^-,\, d_n^+\},\qquad 
   \sdmax = \max\{ d_1^-,\, d_1^+,\ldots, d_n^-,\, d_n^+\}.
\]
In~\cite{digraphs} these are denoted by $\dmin^0$ and $\dmax^0$,
respectively. However, we prefer the above notation 
as we must take powers of the maximum semi-degree. 

\begin{theorem}
Let $\vecdvec = ((d_1^-,d_1^+),\ldots, (d_n^-,d_n^+))$ be a digraphical
directed degree sequence which is switch-irreducible,
such that $\sdmin\geq 1$ and $2\leq \sdmax\leq \frac{1}{4}\, \sqrt{m}$.
Let $\taud$ denote the mixing time of the directed
switch chain $\Md$ with state space $\Omegad$. Then 
\[ \taud \leq  \dfrac{1}{4}\, \sdmax^{16}\, m^{11} \, 
             \left( m\log(m) + \log(\varepsilon^{-1})\right).\]
\label{main-directed}
\end{theorem}
Here, since $m\leq \sdmax n$, the upper bound on $\taud$ given in
Theorem~\ref{main-directed} is an order of $m^2$ larger 
than the bound of $50\, d^{25} n^9$ obtained in the directed 
$d$-regular case~\cite{directed}, where $\sdmin = \sdmax = d$.

The characterisations of 
Berger and M{\" u}ller-Hanneman~\cite{BMH} or LaMar~\cite{lamar} can be applied
to test whether a given directed degree sequence $\vecdvec$
is switch-irreducible. 
Note that $\sdmin \geq 1$ if and only if there are no sources and no
sinks (in any directed graph with degree sequence $\vecdvec$).
Note that the set of 1-regular directed graphs (with $\sdmin=\sdmax =1$)
corresponds to the set of all perfect matchings of $K_{n,n} - M$,
where $M$ is a specified perfect matching (forbidding loops in
the directed graph).  These can be sampled easily, so we assume
that $\sdmax\geq 2$.

We do not believe that the upper bounds given in Theorem~\ref{main} and
Theorem~\ref{main-directed} are tight.  It is likely that the true
mixing time in each case is much lower, perhaps $O(d n \log(dn))$
where $d$ is $\dmax$ or $\sdmax$, respectively. 
Establishing
this appears to be far beyond the reach of known proof techniques.

It is not known whether the corresponding counting problems
(exact evaluation of $|\Omegau|$ or $|\Omegad|$) are $\#P$-complete.
There are several results giving asymptotic enumeration formulae
for $|\Omegau|$, and some for $|\Omegad|$, under various conditions on the 
degree sequence:
see for example~\cite{BH,GMforbidden,ranX,McKW91} and references therein.

The rest of the paper is structured as follows. 
In Section~\ref{ss:history} we review some related results. 
The necessary Markov chain definitions are given in Section~\ref{ss:flow}.  
Then we consider the undirected switch chain
in Section~\ref{s:undirected}, where Theorem~\ref{main} is proved.
Finally, the directed switch chain is studied in 
Section~\ref{s:directed}, where Theorem~\ref{main-directed} is proved. 

\subsection{History and related work}\label{ss:history}

There are several approaches to the problem of sampling graphs 
(or directed graphs) with a given
degree sequence, though none is known to be efficient for all degree sequences.  
First we consider undirected graphs.
The configuration model of Bollob{\' a}s~\cite{bollobas}
gives expected polynomial time uniform sampling if $d_{\max} = O(\sqrt{\log n})$.
McKay and Wormald~\cite{McKW90} adapted the configuration model to give an algorithm which 
performs uniform sampling from
$\Omegau$ in expected polynomial time when $\dmax = O(M^{1/4})$.

Jerrum and Sinclair~\cite{JS90} used a construction of Tutte's~\cite{tutte}
 to reduce the problem
of approximately sampling from $\Omegau$ to the problem of 
approximately sampling perfect matchings
from an auxilliary graph. The resulting Markov chain algorithm is rapidly mixing
if the degree sequence $\dvec$ is \emph{stable}: see~\cite{JMS}.
Stable sequences are those in which small local changes to the degree sequences
do not greatly affect the size of $|\Omegau|$.
Specifically, a graphical degree sequence $\dvec$ is stable if
\[ (\dmax - \dmin + 1)^2 \leq 4 \dmin \left( n - \dmax + 1\right).\]
Many degree sequences which satisfy the conditions of Theorem~\ref{main} are
stable;  however, not all stable sequences satisfy the conditions
of Theorem~\ref{main}. (For example,
if $\dmin = n/9$ and $\dmax =4n/9$ then $\dvec$
is stable~\cite{JMS} but $\sqrt{M}\leq 2n/3$, which is
not large enough for Theorem~\ref{main}.)

We note that Barvinok and Hartigan~\cite{BH} showed that the 
adjacency matrix of a random element of $\Omegau$ is ``close'' 
to a certain ``maximum entropy matrix'', when the degree sequence is \emph{tame}.
The definition of tame depends on the maximum entropy matrix, but a sufficient
condition is that $d_{\min}\geq \alpha (n-1)$ and $d_{\max} \leq \beta (n-1)$
for some constants $\alpha,\beta > 0$.  Some degree sequences satisfying this
latter condition are stable sequences, 
and many of these degree sequences also satisfy the condition of Theorem~\ref{main}.  
It would be interesting to
explore further the connections between stable degree sequences,  tame degree 
sequences and the mixing rate of the switch Markov chain.

Steger and Wormald~\cite{SW} gave an easily-implementable algorithm
for sampling regular graphs, and proved that their algorithm performs
asymptotically uniform sampling in polynomial time when $d=o(n^{1/28})$
(where $d$ denotes the degree).
Kim and Vu~\cite{KV} gave a sharper analysis and established that $d=o(n^{1/3})$
suffices for efficient asymptotically uniform sampling. Bayati, Kim and Saberi~\cite{BKS}
extended Steger and Wormald's algorithm to irregular degree sequences, giving
polynomial-time asymptotically uniform sampling
when $d_{\max} = o(M^{1/4})$.  From this they constructed
a sequential importance sampling algorithm for $\Omegau$. 
A similar approach to that of~\cite{McKW90} was described and analysed
by Zhao~\cite{zhao} in a general combinatorial setting. 
Zhao showed that for sampling from $\Omegau$, 
when $d_{\max} = o(M^{1/4})$, his algorithm performs 
asymptotically uniform sampling in time $O(M)$.

There has been less work on the problem of sampling directed
graphs with a given degree sequence. However, by characterising
directed graphs as bipartite graphs which avoid a certain
perfect matching, it is enough to be able to efficiently sample bipartite
graphs with given degrees. The configuration model can be easily
adapted for bipartite graphs, and gives expected polynomial time 
sampling when
the product of the maximum in-degree and maximum out-degree is
$O(\log n)$, see~\cite{mckay-line}.
Jerrum, Sinclair and Vigoda~\cite{JSV} gave a polynomial-time algorithm
for sampling perfect matchings from a given bipartite graph.
Combining this with Tutte's construction~\cite{tutte} gives a polynomial-time
algorithm for sampling directed graphs with a given degree sequence.
Since the auxilliary graph produced by Tutte's construction has a quadratic
number of vertices,
this method of sampling directed graphs with given degrees
has running time bound $O^\ast(n^{22})$,
where $O^\ast(\cdot)$ ignores logarithmic factors
(by~\cite[Lemma 3.2]{JSV}).
Bayati, Kim and Saberi~\cite{BKS} showed that their sequential
importance sampler could
be adapted for bipartite graphs with given degrees.

For any directed degree sequence $\vecdvec$,
it follows from Rao et al.~\cite{RJB} that the state space
 $\Omegad$ is connected if the set of transitions of the directed
switch chain is expanded to also allow the reversal of directed 3-cycles.
In the bipartite setting, this corresponds to an adaptation of the
chain (for undirected graphs) which sometimes replaces 3 edges per step, rather than 2.
Erd{\H o}s et al.~\cite{EKMS} proved that this chain
is rapidly mixing for half-regular bipartite graphs with a forbidden matching,
where a bipartite graph is half-regular if one vertex bipartition is regular.
This gives an alternative Markov chain for sampling regular directed graphs,
for any degree, including dense regular directed graphs.

We conclude this section with two recent papers.
Erd{\H o}s, Mikl{\' o}s and Toroczkai~\cite{EMT} showed how to build on
the results of~\cite{CDG,SODA,MES}
using several ingredients including 
a Markov chain factorisation theorem by the same authors~\cite{EMT2015} and
a certain canonical decomposition of degree sequences due to 
Tyshkevich~\cite{tysh,tysh2}.
Their approach works by taking degree sequences for which rapid mixing of the switch
chain is known, and combining them in order to construct new degree sequences
for which they prove that the switch chain is also rapidly mixing.
Erd{\H o}s et al.\ also considered the directed setting, where
Theorem~\ref{main-directed} now provides a wider range of directed degree sequences
for which rapid mixing is known, extending the foundation of the method used
in~\cite{EMT}. This should further enlarge the set of more directed degree sequences 
for which the directed switch chain can be shown to be rapidly mixing.

Gao and Wormald~\cite{GW} have recently described an extremely efficient
expected polynomial time algorithm 
for exactly uniform sampling $d$-regular undirected graphs, where $d=o(\sqrt{n})$.
The expected running time of their algorithm is $O(d^3 n)$. 
They also describe a variant of their algorithm with expected running time
$O(dn)$ such that the total variation distance of the output distribution
from uniform is $o(1)$, again when $d=o(\sqrt{n})$.

\subsection{Markov chains and multicommodity flow}\label{ss:flow}

For Markov chain definitions not given here, see for example~\cite{Jerrum}.

Let $\mathcal{M}$ be a Markov chain with finite state space $\Omega$,
transition  matrix $P$ and stationary distribution $\pi$.  The
\emph{total variation distance} between two probability distributions
$\sigma$, $\sigma'$ on $\Omega$ is given by
\[ d_{\mathrm{TV}}(\sigma,\sigma') = \nfrac{1}{2} \sum_{x\in\Omega}
  |\sigma(x) - \sigma'(x)|.
\]
The \emph{mixing time} $\tau(\mathcal{M},\varepsilon)$ is defined by
\[ \tau(\mathcal{M},\varepsilon) = \max_{x\in\Omega}\,
           \min\{ T \geq 0 \mid   d_{\mathrm{TV}}(P^t_x,\pi)\leq \varepsilon
  \,\, \text{ for all $t \geq T$}\}
\]
where $P^t_x$ is the distribution of the state $X_t$ of $\mathcal{M}$ after
$t$ steps from the initial state $X_0=x$.

To bound the mixing time of the switch chain, we apply a multicommodity flow argument.
Suppose that $\mathcal{G}$ is the graph underlying a Markov
chain $\mathcal{M}$, so that $xy$ is an edge of $\mathcal{G}$ if and only
if $P(x,y)>0$. A \emph{flow} in $\mathcal{G}$ is a function 
$f:\mathcal{P}\rightarrow [0,\infty)$ such that
\[ \sum_{p\in\mathcal{P}_{xy}} f(p) = \pi(x)\pi(y) \quad \text{ for all }
 \,\, x,y\in\Omega,\,\, x\neq y.\]
Here $\mathcal{P}_{xy}$ is the set of all simple directed paths from $x$ to
$y$ in $\mathcal{G}$ and $\mathcal{P} = \cup_{x\neq y} \mathcal{P}_{xy}$.
Extend $f$ to a function on oriented edges by setting
$f(e) = \sum_{p\ni e} f(p)$,
so that $f(e)$ is the total flow routed through $e$.  Write $Q(e) = \pi(x) P(x,y)$
for the edge $e=xy$.  Let $\ell(f)$ be the \textsl{length} of the longest path with
$f(p) > 0$, and let $\rho(e) = f(e)/Q(e)$ be the \emph{load} of the edge $e$.
The \emph{maximum load} of the flow is
$\rho(f) = \max_e \rho(e)$.
Using Sinclair~\cite[Proposition 1 and Corollary 6']{sinclair},  the mixing time of
$\mathcal{M}$ can be bounded above by
\begin{equation}
\label{flowbound} 
\tau(\mathcal{M}, \varepsilon) \leq \rho(f)\ell(f)\left(\log(1/\pi^*) + \log(\varepsilon^{-1})\right)
\end{equation}
where $\pi^* = \min\{ \pi(x) \mid x\in\Omega\}$.

\section{The undirected switch chain}\label{s:undirected}

The (undirected) switch Markov chain $\Mu$ has state space
$\Omegau$ and transitions given by the following procedure:
from the current state $G\in\Omegau$, choose an unordered pair of 
two distinct non-adjacent edges uniformly at random, 
say $F=\{\{ x,y\},\{ z,w\}\}$, and choose a perfect matching $F'$
from the set of three perfect matchings of (the complete graph on)
$\{ x,y,z,w\}$, chosen uniformly at random.  If 
$F'\cap \left( E(G) \setminus F\right) = \emptyset$ then the next state
is the graph $G'$ with edge set $\left(E(G) \setminus F\right)\cup F'$,
otherwise the next state is $G'=G$.

Define $M_2 = \sum_{j=1}^n d_j(d_j-1)$. 
If $P(G,G')\neq 0$ and $G\neq G'$ then $P(G,G') = 1/\big(3 a(\dvec)\big)$, where
\begin{equation}
\label{ad}
 a(\dvec) = \binom{M/2}{2} - \dfrac{1}{2}\, M_2 
\end{equation}
is the number of unordered pairs of distinct nonadjacent edges in $G$.
This shows that the switch chain $\Mu$ is symmetric,
and it is aperiodic since $P(G,G)\geq 1/3$ for all
$G\in\Omegau$.

In this section we prove Theorem~\ref{main} by extending the multicommodity
flow argument given in~\cite{CDG} in the regular case.
The definition of the multicommodity flow given in~\cite[Section 2.1]{CDG} 
carries across
to irregular degree sequences without change.  This is because the flow from $G$ to $G'$
depends only on the symmetric difference $G\triangle G'$ of $G$ and $G'$, 
treated as a 2-edge-coloured
graph (with edges from $G\setminus G'$ coloured blue and edges from $G'\setminus G$ 
coloured red, say). 
The blue degree at a given vertex equals the red degree at that vertex, but 
in general the blue degree sequence will not be regular.
Hence the multicommodity flow definition given in~\cite{CDG}
is already general enough to handle irregular degree sequences.

The multicommodity flow is defined using a process which we now sketch.
Given $G, G'\in\Omegau$:
\begin{itemize}
\item Define a bijection from the set of blue edges incident at $v$ to the
set of red edges incident at $v$, for each vertex $v\in [n]$.
The vector of these bijections is called a \emph{pairing} $\psi$,
and the set of all possible pairings is denoted $\Psi(G,G')$.
\item  The pairing gives a canonical way to decompose the symmetric difference
$G\triangle G'$ into a sequence of \emph{circuits}, where each circuit is a
blue/red-alternating closed walk.
\item Each circuit is decomposed in a canonical way
into a sequence of simpler circuits
of two types: 1-\emph{circuits} and 2-\emph{circuits}.  A 1-circuit is
an alternating cycle in $G\triangle G'$, while a 2-circuit is
an alternating walk with one vertex of degree 4, the rest of degree 2,
consisting of two odd cycles which share a common vertex.
Each 1-circuit or 2-circuit has a designated \emph{start vertex}.
(The start vertex of a 2-circuit is the unique vertex of degree 4.)
An important fact is that the 1-circuits and 2-circuits are \emph{pairwise 
edge-disjoint}.
\item  Each 1-circuit or 2-circuit is processed in a canonical way to
give a segment of the canonical path $\gamma_\psi(G,G')$.
\end{itemize}
Thus, for each $(G,G')\in\Omegau^2$ and each $\psi\in\Psi(G,G')$,
we define a (canonical) path $\gamma_{\psi}(G,G')$ from $G$ to $G'$.
For full details see~\cite[Section 2.1]{CDG}.

Next, the value of the flow along this path is defined as follows:
\begin{equation} f(\gamma_{\psi(G,G')}) = \frac{1}{|\Omegau|^2\, |\Psi(G,G')|}
\label{f-def}
\end{equation}
and setting $f(p)=0$ for any other directed path from $G$ to $G'$.
Recall that $\mathcal{P}_{G,G'}$ is defined to be the set of all
directed paths from $G$ to $G'$, in the underlying digraph of $\Mu$.
Summing $f(p)$ over all $p\in\mathcal{P}_{G,G'}$ gives
$1/|\Omegau|^2 = \pi(G)\pi(G')$, as required for a valid flow.
This flow from $G$ to $G'$ has been equally shared among all paths in
$\{\gamma_{\psi}(G,G')\mid \psi\in\Psi(G,G')\}$.

\subsection{Analysing the flow}\label{s:analysis}

Now we show how to bound the load of the flow by adapting the analysis 
from~\cite{CDG}. 
Note that some proofs in~\cite{CDG} used the assumption $d=d(n)\leq n/2$,
since the general result follows by complementation.  This trick does
not work for irregular degree sequences, so we cannot make a similar
assumption here.

Given matrices $G$, $G'$, $Z\in\Omegau$, define the
\emph{encoding} $L$ of $Z$ (with respect to $G,G'$) by
\[ L + Z = G + G'\]
by identifying each of $Z$, $G$ and $G'$ with their symmetric 0-1 adjacency matrices.
Then $L$ is a symmetric $n\times n$ matrix with entries in $\{ -1, 0, 1,2\}$
and with zero diagonal.
Entries which equal $-1$ or 2 are called \emph{defect entries}.  Treating $L$ as an
edge-labelled graph with edges labelled $-1, 1, 2$ (and omitting edges corresponding to
zero entries), a \emph{defect edge} is an edge labelled $-1$ or $2$.
(In~\cite{CDG} these were called ``bad edges''.)
Specifically, we will refer to $(-1)$-\emph{defect edges} and to 
$2$-\emph{defect edges}.  A 2-defect edge is present in both $G$ and
$G'$ but is absent in $Z$, while a $(-1)$-defect edge is absent in
both $G$ and $G'$ but is present in $Z$.

We say that the \emph{degree} of vertex $v$ in $L$ is the sum of the
labels of the edges incident with $v$ (equivalently, the sum of the
entries in the row of $L$ corresponding to $v$).  By definition,
the degree sequence of $L$ equals $\dvec$.

Some proofs from~\cite{CDG,CDG-corrigendum} also apply in the
irregular case without any substantial change (after replacing
$d$ by $\dmax$). These proofs refer only to the symmetric difference 
and the process used to construct the multicommodity
flow (and none of them use the assumption $d\leq n/2$).  
 We state two of these results now. 

\begin{lemma}
Suppose that $G,G',Z,Z'\in\Omegau$ are such that $(Z,Z')$ is
a transition of the switch chain which lies on the canonical path 
$\gamma_\psi(G,G')$ for some $\psi\in\Psi(G,G')$.
Let $L$ be the encoding of $Z$ with respect to $(G,G')$.  Then the following statements
hold:
\begin{enumerate}
\item[\emph{(i)}]
\emph{(\cite[Lemma 1]{CDG})}\ From $(Z,Z')$, $L$ and $\psi$ it is possible to uniquely recover $G$ and $G'$.
\item[\emph{(ii)}] 
\emph{(\cite[Lemma 2]{CDG})}\
There are at most four defect edges in $L$.  The labelled
graph consisting of the defect edges in $L$ must form a subgraph of one
of the five possible labelled graphs shown in Figure~\ref{f:possible},
where \emph{``?''} represents a label which may be either $-1$ or $2$.
\end{enumerate}
\label{oldstuff}
\end{lemma}

\begin{proof}
The proof of these results for the regular case,~\cite[Lemma 1, Lemma 2]{CDG}, also 
applies here since we are using the same multicommodity flow definition, and the proof only 
uses
the symmetric difference of $G$ and $G'$.  (The assumption that $d\leq n/2$, which was
sometimes made in~\cite{CDG}, is not used in the proof of~\cite[Lemma 1, Lemma 2]{CDG}.)
\end{proof}

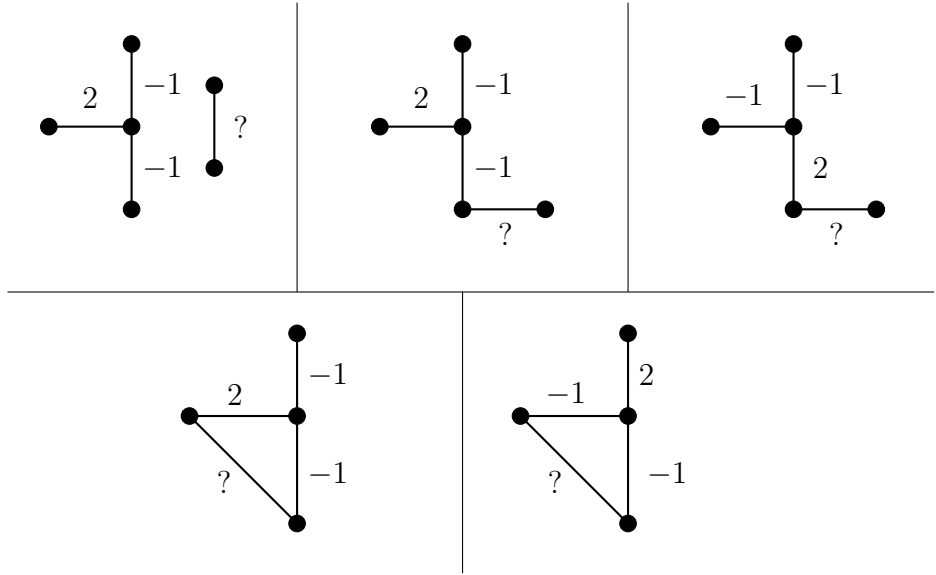
\begin{figure}[ht!]
\begin{center}
\begin{tikzpicture}[scale=1.1]
\draw [-,thick] (1.0,3) -- (1.0,5);
\draw [-,thick] (-0.0,4) -- (1.0,4);
\draw [-,thick]  (2.0,3.5) -- (2.0,4.5);
\draw [fill] (1.0,3) circle (0.1); \draw [fill] (1.0,5) circle (0.1);
\draw [fill] (1.0,4) circle (0.1); \draw [fill] (0.0,4) circle (0.1);
\draw [fill] (2.0,3.5) circle (0.1); \draw [fill] (2.0,4.5) circle (0.1);
\node [above] at (0.5,4.1) {$2$};
\node [right] at (1.0,4.5) {$-1$};
\node [right] at (1.0,3.5) {$-1$};
\node [right] at (2.1,4.0) {$?$};
\draw [-,thick] (4,4) -- (5,4);
\draw [-,thick] (5,5) -- (5,3) -- (6,3);
\draw [fill] (4,4) circle (0.1); \draw [fill] (5,4) circle (0.1);
\draw [fill] (5,5) circle (0.1); \draw [fill] (5,3) circle (0.1);
\draw [fill] (6,3) circle (0.1); 
\node [above] at (4.5,4.1) {$2$};
\node [right] at (5.0,4.5) {$-1$};
\node [right] at (5.0,3.5) {$-1$};
\node [right] at (5.3,2.7) {$?$};
\draw [-,thick] (8.0,4) -- (9.0,4);
\draw [-,thick] (9.0,5) -- (9.0,3) -- (10.0,3);
\draw [fill] (8.0,4) circle (0.1); \draw [fill] (9.0,4) circle (0.1);
\draw [fill] (9.0,5) circle (0.1); \draw [fill] (9.0,3) circle (0.1);
\draw [fill] (10.0,3) circle (0.1); 
\node [above] at (8.4,4.1) {$-1$};
\node [right] at (9.0,4.5) {$-1$};
\node [right] at (9.1,3.5) {$2$};
\node [right] at (9.3,2.7) {$?$};
\draw [-,thick] (3,1.5) -- (3,-0.8) -- (1.7,0.5) -- (3,0.5);
\draw [fill] (3,1.5) circle (0.1); \draw [fill] (3,-0.8) circle (0.1);
\draw [fill] (1.7,0.5) circle (0.1); \draw [fill] (3,0.5) circle (0.1);
\node [above] at (2.25,0.5) {$2$};
\node [right] at (3.0,-0.2) {$-1$};
\node [right] at (3.0,1.0) {$-1$};
\node [right] at (1.9,-0.3) {$?$};
\draw [-,thick] (7,1.5) -- (7,-0.8) -- (5.7,0.5) -- (7,0.5);
\draw [fill] (7,1.5) circle (0.1); \draw [fill] (7,-0.8) circle (0.1);
\draw [fill] (5.7,0.5) circle (0.1); \draw [fill] (7,0.5) circle (0.1);
\node [above] at (6.25,0.5) {$-1$};
\node [right] at (7.1,-0.2) {$-1$};
\node [right] at (7.0,1.0) {$2$};
\node [right] at (5.9,-0.3) {$?$};
\draw [-] (-0.5,2) -- (10.7,2);
\draw [-] (5,2) -- (5,-1.4);
\draw [-] (7,2) -- (7,5.5);
\draw [-] (3,2) -- (3,5.5);
\end{tikzpicture}
\caption{The five possible configurations of four defect edges}
\label{f:possible}
\end{center}
\end{figure}

The next result collects together some further useful results about
encodings.  

\begin{lemma}
Suppose that the conditions of Lemma~\ref{oldstuff} hold.
Let $x,y,z\in [n]$ be distinct vertices.  
\begin{itemize}
\item[\emph{(i)}] If $L(x,y)=2$ then $d_x\geq 2$, $d_y\geq 2$.  
\item[\emph{(ii)}] If $L(x,y)=2$ and $L(y,z)=2$ then $d_y\geq 4$.
\item[\emph{(iii)}]  If $L(x,y)=2$ and $L(y,z)=-1$ then $d_y\geq 3$.
\end{itemize}
\label{structure}
\end{lemma}
\begin{proof}
It follows from the definition of the multicommodity flow given in~\cite{CDG}
that a 2-defect edge $\{ x,y\}$ (with $L(x,y)=2$) can only arise in
two cases: 
\begin{itemize}
\item[(a)]
$\{x,y\}$ is a \emph{shortcut edge} 
which is present in $G, G'$ but which is
absent in $Z$.  (See~\cite[Figure 4]{CDG}.)
In this case, $x$ and $y$ are vertices on some
2-circuit, which is an alternating blue/red walk in the symmetric difference
$G\triangle G'$.  Hence both $x$ and $y$ have degree at least two
in $G$.
\item[(b)] $\{ x,y\}$ is an \emph{odd chord} which is present in $G, G'$
but which is absent in $Z$.    
(See the section ``Processing a 1-circult'' in~\cite{CDG}.) 
In this case, $x$ and $y$ are vertices
on some 1-circuit, which is an alternating blue/red walk in the
symmetric difference $G\triangle G'$.  
Again, this shows that both $x$ and $y$ have degree at least two in $G$.
\end{itemize}
This proves (i).  

Next, if $y$ is incident with two edges
of defect 2 then it must be that 
one is an odd chord for a 1-circuit $C_1$ and
one is a shortcut edge for a 2-circuit $C_2$,
 where $y$ does not play the
role of $x_0$ for $C_1$.  
Then $y$ is incident in $G$ with an edge of $C_1$, an edge of $C_2$
and the two edges $\{ x,y\}$, $\{y,z\}$ which are 2-defect edges in $L$.
Since $C_1$ and $C_2$ are edge-disjoint and no defect edge belongs
to $G\triangle G'$, it follows that $d_y\geq 4$, proving (ii).

We may adapt this argument to prove (iii), 
noting that a $(-1)$-defect edge may only
arise from a shortcut edge or an odd chord which is
absent in $G$ and $G'$ and present in $Z$.  
\end{proof}

We say that an encoding $L$ is \emph{consistent} with $Z$ 
if $L+Z$ only takes entries in $\{0,1,2\}$.
Say that an encoding is \emph{valid} if it satisfies the conclusions
of Lemma~\ref{oldstuff}(ii), and that a valid encoding is \emph{good}
if it also satisfies the conclusion of Lemma~\ref{structure}. 
Let $\mathcal{L}(Z)$ be the set of valid encodings which are consistent with
$Z$, and let $\mathcal{L}^*(Z)$ be the set of good encodings which are
consistent with $Z$.  In~\cite{CDG} the set $\mathcal{L}(Z)$ was
studied, but we can obtain a slightly better upper bound if we
work with the smaller set $\mathcal{L}^*(Z)$.

\begin{lemma}
\emph{(\cite[Lemma 5]{CDG} and~\cite[Lemma 1]{CDG-corrigendum})}\
The load $f(e)$ on the transition $e=(Z,Z')$ satisfies
\[ f(e)\leq  \dmax^{14}\,{\frac{|\mathcal{L}^*(Z)|}{|\Omegau|^2}}.\]
\label{fbound}
\end{lemma}
\begin{proof}
We adapt the proof given for the regular case in~\cite[Lemma 5]{CDG} 
and~\cite[Lemma 1]{CDG-corrigendum}, noting that this proof did not
require the assumption $d\leq n/2$.
An outline of the argument is given below.

For a given transition $e=(Z,Z')$, recall that $f(e)$ is the sum of
$f(p)$ over all paths $p\in\mathcal{P}$ which contain $e$.
That is,
\[
 |\Omegau|^2 \, f(e) = \sum_{(G,G')}\, \sum_{\psi\in\Psi(G,G')}\,
  \mathbf{1}(e\in\gamma_\psi(G,G'))\, \frac{1}{|\Psi(G,G')|},
\]
using (\ref{f-def}) and the definition of $f(e)$. Now, by Lemma~\ref{oldstuff}(i), each pair
$(G,G')$ such that $e\in\gamma_{\psi(G,G')}$ can be uniquely
reconstructed from the encoding $L\in\mathcal{L}^*$ defined by $L+Z = G+G'$.   
Furthermore,
given $Z$ and $L$ we can construct a ``yellow-green'' colouring of the
symmetric difference $L\triangle Z = G\triangle G'$,
where yellow edges have label 1 under $L$ and do not occur in $Z$, and green
edges occur in $Z$ and have label 0 under $L$.  Suppose that $e\in\gamma_\psi(G,G')$ for
some pair of states $(G,G')$ and some pairing $\psi$.  
In~\cite{CDG-corrigendum},
using Lemma~\ref{oldstuff}(ii), it was proved that $\psi$ will pair yellow
edges to green edges almost everywhere, with at most 14 ``bad pairs'' where
$\psi$ maps yellow to yellow, or green to green.  
This proof also holds in the irregular setting.

Therefore, if $\Psi'(L)$ is the set of pairings of $L\triangle Z$ with
at most 14 bad pairs then
\[  |\Omegau|^2\, f(e)    \leq \sum_{L\in\mathcal{L}^\ast(Z)}\, \sum_{\psi\in\Psi'(L)} \,
  \,  \frac{1}{|\Psi(G,G')|},
\]
since each $(L,\psi)\in \mathcal{L}^*(Z)\times \Psi'(L)$ contributes at most one pair $(G,G')$
with $e\in \gamma_{\psi}(G,G')$, and each such pair $(G,G')$ is included in the
sum over $(L,\psi)$.
Finally, in~\cite{CDG-corrigendum} it was shown that, in the regular setting
\[ |\Psi'(L)|\leq \dmax^{14}\, |\Psi(G,G')|
\]
for any $L\in\mathcal{L}^*(Z)$.  The same argument works for irregular degree
sequences, noting that the original argument did not use the condition $d\leq n/2$.
Combining these last two displayed equations proves that
\[ |\Omegau|^2\, f(e) \leq \dmax^{14}\, |\mathcal{L}^*(Z)|,
\]
as required.
\end{proof}

The switch operation can be extended to encodings 
in the natural way: each switch reduces two
edge labels by one and increases two edge labels by one, without changing the degrees.
It was shown in~\cite[Lemma 3]{CDG} that from any valid encoding,
one could obtain a graph (with no defect edges)
by applying a sequence of at most three switches.  
In~\cite[Lemma 4]{CDG} we used this fact to bound the ratio
$|\mathcal{L}(Z)|/|\Omegau|$ for regular degree sequences.  
This provided an upper bound for the flow $f(e)$ through the transition $e=(Z,Z')$
(as in Lemma~\ref{fbound}, above).  Recall that we now seek an
upper bound on the slightly smaller ratio $|\mathcal{L}^*(Z)|/|\Omegau|$.

The proof of~\cite[Lemma 3]{CDG} uses regularity to prove 
the existence of certain edges exist which are
needed in order to find switches to remove the defect edges.
This argument fails for irregular degree sequences.  
However, any argument which gives an
upper bound on $|\mathcal{L}^*(Z)|/|\Omegau|$ will do. 
So we will instead consider
a slightly more complicated operation than a switch, which we call a 
\emph{3-switch} (this operation is called a
``circular $C_6$-swap'' in~\cite{EKMS}).  
(This approach of considering
more complicated operations in order to obtain more freedom has been used to improve
asymptotic enumeration results, for example in~\cite{McKW91}.)
We remark that this new operation is only used to give an upper bound 
on the ratio $|\mathcal{L}^*(Z)|/|\Omegau|$, and is not related to the switches
performed by the Markov chain $\Mu$.


A 3-switch is described by a 6-tuple $(a_1,b_1,a_2,b_2,a_3,b_3)$ of distinct vertices
such that $\{a_1,b_1\}$, $\{a_2,b_2\}$, $\{a_3,b_3\}$ are all edges and $\{a_2,b_1\}$, $\{a_3,b_2\}$, $\{a_1,b_3\}$ are
all non-edges. The 3-switch deletes the three edges $\{a_1,b_1\}$, $\{a_2,b_2\}$, $\{a_3,b_3\}$ from the
edge set and replaces them with $\{a_2,b_1\}$, $\{a_3,b_2\}$, $\{a_1,b_3\}$, as shown in Figure~\ref{3-switch}.

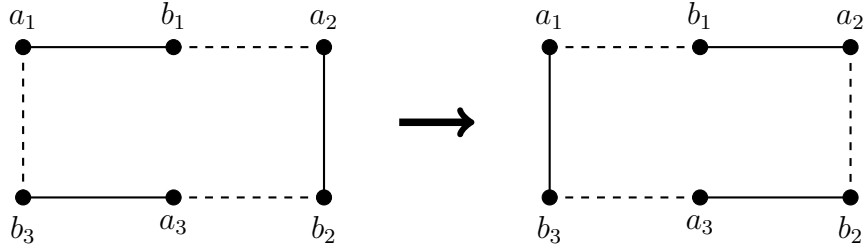
\begin{figure}[ht!]
\begin{center}
\begin{tikzpicture}
\draw [-,thick,dashed] (0,0) -- (0,2);
\draw [-,thick,dashed] (2,2) -- (4,2);
\draw [-,thick,dashed] (4,0) -- (2,0);
\draw [-,thick] (0,2) -- (2,2);
\draw [-,thick] (4,2) -- (4,0);
\draw [-,thick] (2,0) -- (0,0);
\draw [fill] (0,0) circle (0.1);
\draw [fill] (2,0) circle (0.1);
\draw [fill] (4,0) circle (0.1);
\draw [fill] (0,2) circle (0.1);
\draw [fill] (2,2) circle (0.1);
\draw [fill] (4,2) circle (0.1);
\node [below] at (0,-0.1) {$b_3$};
\node [above] at (0.0,2.1) {$a_1$};
\node [above] at (2,2.1) {$b_1$};
\node [above] at (4,2.1) {$a_2$};
\node [below] at (4,-0.1) {$b_2$};
\node [below] at (2.0,-0.1) {$a_3$};
\draw [->,line width = 1mm] (5,1) -- (6,1);
\draw [-,thick] (7,0) -- (7,2);
\draw [-,thick] (9,2) -- (11,2);
\draw [-,thick] (11,0) -- (9,0);
\draw [-,thick,dashed] (7,2) -- (9,2);
\draw [-,thick,dashed] (11,2) -- (11,0);
\draw [-,thick,dashed] (9,0) -- (7,0);
\node [below] at (7,-0.1) {$b_3$};
\node [above] at (7,2.1) {$a_1$};
\node [above] at (9,2.1) {$b_1$};
\node [above] at (11,2.1) {$a_2$};
\node [below] at (11,-0.1) {$b_2$};
\node [below] at (9,-0.1) {$a_3$};
\draw [fill] (7,0) circle (0.1);
\draw [fill] (7,2) circle (0.1);
\draw [fill] (9,2) circle (0.1);
\draw [fill] (11,2) circle (0.1);
\draw [fill] (11,0) circle (0.1);
\draw [fill] (9,0) circle (0.1);
\end{tikzpicture}
\caption{A 3-switch}
\label{3-switch}
\end{center}
\end{figure}
The 3-switch can also be extended to encodings.

Let $C(p,q)$ be the set of encodings in $\mathcal{L}^\ast(Z)$ with 
precisely $p$ defect edges labelled 2 and precisely $q$ defect edges
labelled $-1$, 
for $p\in\{0,1,2\}$ and $q\in\{0,1,2,3\}$. 
Then $\Omegau= C(0,0)$
and 
\[ \mathcal{L}^\ast(Z) = \bigcup_{p=0}^2\,\bigcup_{q=0}^3\, C(p,q),\] 
where this union is disjoint.
(Note that $C(2,3)=\emptyset$, by Lemma~\ref{oldstuff}(ii).)

If $L\in C(p,q)$ then there
are precisely $M/2 - 2p + q$ non-defect edges in $L$.  (To see this,
note that the sum of all entries in the matrix $L$ must equal $M$,
and $L$ has zero diagonal.)  

For $v\in [n]$, given an encoding $L$,  write $N_L(v)$ to denote
the set of $w\in [n]\setminus \{v\}$ such that $L(v,w)=1$,
and let $\hatN_L(v)$ be the set of all $w\in [n]\setminus \{v\}$ such that
$L(v,w)\neq 0$.  Then $N_L(v)$ is the set of neighbours of $v$ in $L$,
where neighbours along defect edges are not included, while $\hatN_L(v)$ is
the set of all neighbours of $v$ in $L$ (counting both defect and non-defect
edges).  Let $\zeta_v$ (respectively $\eta_v$) denote
the number of 2-defect edges (respectively, ($-1$)-defect edges) incident
with $v$ in $L$.  Then 
\begin{equation}
\label{good-edges}
 |N_L(v)| = d_v - 2\zeta_v + \eta_v,
\end{equation}
and hence the total number of edges incident with $v$ in $L$
is 
\begin{equation}
\label{all-edges} |\hatN_L(v)| = |N_L(v)| + \zeta_v + \eta_v = 
   d_v - \zeta_v + 2\eta_v.
\end{equation}

With these formulae we can prove the following bounds, which will be
very useful in our arguments.

\begin{lemma}
\label{useful-bounds}
Suppose that $L\in C(p,q)$ and let $a_1,b_1$ be distinct vertices with $L(a_1,b_1)\neq 0$.
\begin{itemize}
\item[\emph{(i)}] The number of ways to choose an ordered pair of vertices
$(a_2,b_2)$ such that $L(a_2,b_2)=1$ and $L(a_2,b_1)=0$, with $a_1,b_1,a_2,b_2$ all distinct,
is at least 
\begin{align*} M - 4p + 2q -
 & \Bigg( \dmax\Big(\dmax -\zeta_{b_1} + 2\eta_{b_1} + 2\Big) + \eta_{a_1} + \eta_{b_1} 
-  2(\zeta_{a_1} + \zeta_{b_1})
 \\ & \hspace*{6cm} {} 
+ \sum_{y\in \hatN_L(b_1)} (\eta_y - 2\zeta_y)\Bigg).
\end{align*}
\item[\emph{(ii)}]
Now suppose that $a_1$, $b_1$, $a_2$, $b_2$ are distinct vertices with $L(a_2,b_2)=1$.
Define
\[ \eta^\ast = \eta_{a_1} + \eta_{b_1} + \eta_{a_2} + \eta_{b_2},\qquad
   \zeta^\ast = \zeta_{a_1} + \zeta_{b_1} + \zeta_{a_2} + \zeta_{b_2}.
\]
The number of ways to choose an ordered pair of vertices $(a_3,b_3)$ such that
$L(a_3,b_3)=1$ and $L(a_1,b_3)=L(a_3,b_2)=0$, with $a_1,b_1,a_2,b_2,a_3,b_3$ all distinct, 
is at least
\begin{align*}
M - 4p+2q - & \Bigg(\dmax\Big(2\dmax - (\zeta_{a_1} +\zeta_{b_2}) + 2(\eta_{a_1} + \eta_{b_2}) + 4\Big)
  + \eta^\ast - 2\zeta^\ast 
 \\ & \hspace*{2cm} {} 
   + \sum_{x\in\hatN_L(a_1)} (\eta_x - 2\zeta_x) + \sum_{y\in\hatN_L(b_2)} (\eta_y - 2\zeta_y) \Bigg).
\end{align*}
\end{itemize}
\end{lemma}

\begin{proof}
For (i), there are $M-4p+2q$ possibilities for $(a_2,b_2)$ with $L(a_2,b_2)=1$,
but we must reject the following choices of $(a_2,b_2)$:
\begin{itemize}
\item those for which $a_1,b_1,a_2,b_2$ are not distinct (that is, any
choice of $(a_2,b_2)$ which is incident with one of $a_1$ or $b_1$), and
\item those with $L(a_2,b_1)\neq 0$. 
\end{itemize}
(In~\cite{SODA}, we neglected to rule out the possibility that there
was a defect edge present between $b_1$ and $a_2$.  We plug this gap here.)
We claim that the number of bad choices for $(a_2,b_2)$ is at most
\begin{equation}
\label{claim}
 \left( \sum_{y\in \hatN_L(b_1)} |N_L(y)|\right) +
    |N_L(a_1)|+ |N_L(b_1)|. 
\end{equation}
To see this, observe that the sum over $y$ counts all ordered pairs $(y,z)$ with $L(y,z)=1$ and
$L(y,b_1)\neq 0$.  This includes each non-defect edge incident with $b_1$
(when $z=b_1)$ and each non-defect edge incident with $a_1$
(when $y=a_1$).   If $L(a_1,b_1)=1$ then this choice is counted twice, which covers
both $(a_1,b_1)$ and $(a_2,b_2)$.  Indeed, all edges incident with $a_1$ or $b_1$ must
be counted twice,
to account for the two choices of orientation of $\{ a_2,b_2\}$.
We achieve this by adding $|N_L(a_1)| + |N_L(b_1)|$ to the upper bound.
All other edges counted by the sum over $y$ are of the form $(y,z)$
with 
\[ y\in\hatN_L(b_1)\setminus \{ a_1\},\qquad z\in N_L(y)\setminus \{b_1\} \,\,\, \text{ and } \,\,\, L(y,z)=1,\]
corresponding a bad choice of $(a_2,b_2)=(y,z)$.
This covers all bad choices of $(a_2,b_2)$, completing the proof of (\ref{claim}).

Applying (\ref{good-edges}) to (\ref{claim}), it follows that the number of bad choices
for $(a_2,b_2)$ is at most
\begin{align*}
 \Bigg( \sum_{y\in \hatN_L(b_1)} (\dmax  - 2\zeta_y + \eta_y) \Bigg) +
    2\dmax - 2(\zeta_{a_1} + \zeta_{b_1}) + \eta_{a_1} + \eta_{b_1}. 
\end{align*}
Applying (\ref{all-edges}) to $b_1$ shows that the number of bad choices for
$(a_2,b_2)$ is at most
\begin{align}
 \dmax\Big(\dmax - \zeta_{b_1} + 2\eta_{b_1}\Big) + 
    2\dmax & - 2(\zeta_{a_1} + \zeta_{b_1}) + \eta_{a_1} + \eta_{b_1}\nonumber \\
 & \hspace*{4cm} {} + \Bigg( \sum_{y\in \hatN_L(b_1)} (\eta_y - 2\zeta_y ) \Bigg), \label{bad-for}
\end{align} 
and subtracting this expression from $M-4p+2q$ completes the proof of (i).

The bound for (ii) is established in a similar fashion.  There are
$M-4p+2q$ choices for $(a_3,b_3)$ such that $L(a_3,b_3)=1$.  An upper bound
on the number of bad choices for $(a_3,b_3)$ is obtained by summing the
upper bound for the number of choices which are bad with respect to
the pair $(a_1,b_1)$, as given by (\ref{bad-for}), and the number of
choices which are bad with respect to the pair $(a_2,b_2)$. The latter is also
given by (\ref{bad-for}) after replacing $a_1$ by $a_2$ and $b_1$ by $b_2$,
and (for clarity) using the dummy variable $x$ in the sum, rather than $y$.
The proof of (ii) is then completed by subtracting the sum of these
two upper bounds from $M-4p+2q$, using the definition of $\eta^*$, $\zeta^*$.
\end{proof}

The following lemma is the ``critical lemma''
which relied on regularity in~\cite{CDG}; its proof is the main task of this
section.

\begin{lemma}
Suppose that $\dmin\geq 1$ and $3\leq \dmax\leq \nfrac{1}{3}\sqrt{M}$.
Let $Z\in\Omegau$.  Then 
\[ |\mathcal{L}^\ast(Z)| \leq  2\, M^6\,  |\Omegau|.\]
\label{3switches}
\end{lemma}
\begin{proof}
We prove that any $L\in\mathcal{L}^\ast(Z)$ can be transformed into
an element of $\Omegau$ (with no defect edges) using a sequence
of at most three 3-switches.  The strategy is as follows:  in Phase 1 we
aim to remove two defects per 3-switch (one 2-defect edge and one $(-1)$-defect edge),
then in Phase 2 we remove one 2-defect edge per 3-switch, and finally in Phase 3
we remove one $(-1)$-defect edge per 3-switch.
There is at most one step in Phase 1, though the other phases may have more than
one step: any phase may be empty.
Each 3-switch we perform gives rise to an upper bound on certain ratios
of the sizes of the sets $C(p,q)$, by double counting.  
The proof is completed by combining these bounds.
(Such an argument is often called a ``switching argument'' in the
asymptotic enumeration literature: see~\cite{McKW91} for example.)

\bigskip

\noindent {\bf Phase 1.}\ 
If $p+q\leq 3$ then Phase 1 is empty: proceed to Phase 2. Otherwise, suppose that
$L\in C(p,q)$ where $p+q=4$, so $(p,q)\in \{ (2,2),\, (1,3)\}$. 
(Recall that there are at most 4 defect edges, by Lemma~\ref{oldstuff}(ii).)
We count the number of 3-switches $(a_1,b_1,a_2,b_2,a_3,b_3)$ 
which
may be applied to $L$ to produce an encoding $L'\in C_{p-1,q-1}$.
This operation is shown in Figure~\ref{double-switch}, where defect
edges are labelled by $2$ or $-1$ and are shown using thicker lines: 
a thick solid line is a 
2-defect edge while a thick dashed line is a $(-1)$-defect edge.
\begin{figure}[ht!]
\begin{center}
\begin{tikzpicture}
\draw [-,thick,dashed] (0,0) -- (0,2);
\draw [-,line width=2.5pt,dashed] (2,2) -- (4,2);
\draw [-,thick,dashed] (4,0) -- (2,0);
\draw [-,line width=2.5pt] (0,2) -- (2,2);
\node [above] at (1,2.1) {$2$};
\node [above] at (3,2.1) {$-1$};
\draw [-,thick] (4,2) -- (4,0);
\draw [-,thick] (2,0) -- (0,0);
\draw [fill] (0,0) circle (0.1);
\draw [fill] (2,0) circle (0.1);
\draw [fill] (4,0) circle (0.1);
\draw [fill] (0,2) circle (0.1);
\draw [fill] (2,2) circle (0.1);
\draw [fill] (4,2) circle (0.1);
\node [below] at (0,-0.1) {$b_3$};
\node [above] at (0.0,2.1) {$a_1$};
\node [above] at (2,2.1) {$b_1$};
\node [left] at (-0.5,1) {$L$};
\node [above] at (4,2.1) {$a_2$};
\node [below] at (4,-0.1) {$b_2$};
\node [below] at (2.0,-0.1) {$a_3$};
\draw [->,line width = 1mm] (5,1) -- (6,1);
\draw [-,thick] (7,0) -- (7,2);
\draw [-,thick,dashed] (9,2) -- (11,2);
\draw [-,thick] (11,0) -- (9,0);
\draw [-,thick] (7,2) -- (9,2);
\draw [-,thick,dashed] (11,2) -- (11,0);
\draw [-,thick,dashed] (9,0) -- (7,0);
\node [below] at (7,-0.1) {$b_3$};
\node [above] at (7,2.1) {$a_1$};
\node [above] at (9,2.1) {$b_1$};
\node [above] at (11,2.1) {$a_2$};
\node [right] at (11.5,1) {$L'$};
\node [below] at (11,-0.1) {$b_2$};
\node [below] at (9,-0.1) {$a_3$};
\draw [fill] (7,0) circle (0.1);
\draw [fill] (7,2) circle (0.1);
\draw [fill] (9,2) circle (0.1);
\draw [fill] (11,2) circle (0.1);
\draw [fill] (11,0) circle (0.1);
\draw [fill] (9,0) circle (0.1);
\end{tikzpicture}
\caption{A 3-switch in Phase 1}
\label{double-switch}
\end{center}
\end{figure}
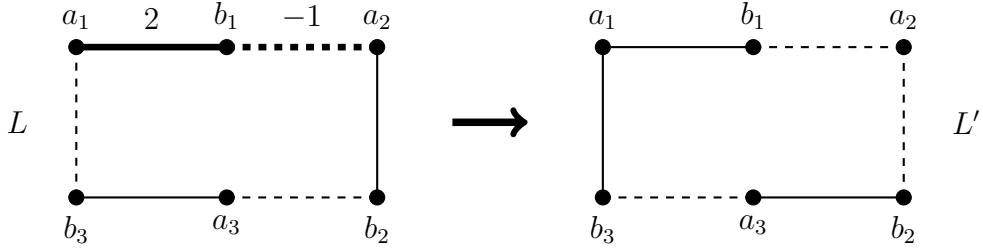

It follows from Figure~\ref{f:possible}
that there are at least two choices for a triple of distinct vertices
$(a_1,b_1,a_2)$ such that 
$L(a_1,b_1) = 2$ and $L(a_2,b_1) = -1$.  

Given $a_1, b_1, a_2$, there
is at least one vertex $b_2\in N_L(a_2)\setminus \{ a_1\}$.
To see this, first suppose that $a_2$ is not incident with a 2-defect edge.
Then $N_L(a_2)$ has at least $d_{a_2}+1\geq 2$ elements, leaving at least
one which is distinct from $a_1$.  Otherwise, if $a_2$ is incident
with a 2-defect edge then it can be incident with at most one 2-defect edge,
since $p\leq 2$. Then there are at least $d_{a_2}-2$ choices for $b_2$
in $N_L(a_2)\setminus \{ a_1\}$,
and this number is positive by Lemma~\ref{structure}(iii).

Next, we choose $(a_3,b_3)$ such that all six vertices are distinct,
$L(a_3,b_3)=1$ and $L(a_1,b_3) = L(a_3,b_2)=0$.
A lower bound for the number of ways to choose $(a_3,b_3)$ is given
in Lemma~\ref{useful-bounds}(ii).  
In this expression, the worst case is obtained by making
$\eta_{a_1}+\eta_{b_2}$ as large as possible, then making
$\eta^\ast$ as large as possible, while making $\zeta^*$ as small as possible.
Additionally, adding an edge (with label 1) between some of the known vertices
may increase the sum over $x\in \hatN_L(a_1)$ or the sum over $y\in \hatN_L(b_2)$.  
A worst-case example is shown in Figure~\ref{worstcase}.
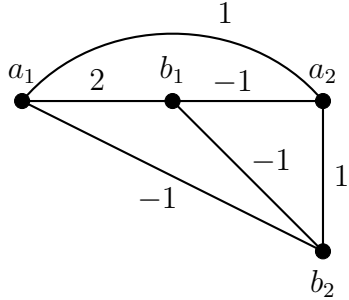
\begin{figure}[ht!]
\begin{center}
\begin{tikzpicture}
\draw[fill] (0,2) circle (0.1);
\draw[fill] (2,2) circle (0.1);
\draw[fill] (4,0) circle (0.1);
\draw[fill] (4,2) circle (0.1);
\draw [thick,-] (0,2) to [out=50,in=130] (4,2);
\node [below] at (1.8,1.0) {$-1$};
\node [above] at (2.8,1.95) {$-1$};
\node [above] at (2.7,2.9) {$1$};
\node [right] at (2.9,1.2) {$-1$};
\node [right] at (4.0,1.0) {$1$};
\node [above] at (1,2) {$2$};
\draw[thick,-] (4,2) -- (4,0);
\draw[thick,-] (4,2.00) -- (2,2.00);
\draw[thick,-] (0,2) -- (4,0);
\draw[thick,-] (0,2) -- (2,2);
\draw[thick,-] (2,2) -- (4,0);
\node [above] at (0,2.1) {$a_1$};
\node [above] at (2,2.1) {$b_1$};
\node [above] at (4,2.1) {$a_2$};
\node [below] at (4,-0.1) {$b_2$};
\end{tikzpicture}
\caption{A worst case configuration for the choice of $(a_3,b_3)$ in Phase 1.}
\label{worstcase}
\end{center}
\end{figure}

\noindent Here $(p,q)=(1,3)$ and
\begin{align*}
 \eta^\ast = 6, \quad \zeta^\ast = 2,\quad &\eta_{a_1} + \eta_{b_2} = 3,\quad \zeta_{a_1} + \zeta_{b_2} = 1,\\
 \sum_{x\in\hatN_L(a_1)} (\eta_x - 2\zeta_x) = 3, \qquad & \sum_{y\in \hat{N}_L(b_2)} (\eta_y - 2\zeta_y) = 0.
\end{align*}
Substituting these values into Lemma~\ref{useful-bounds}(ii) shows that the 
number of good choices for $(a_3,b_3)$ is at least
\begin{align*}
M + 2 - (2\dmax^2 + 9\dmax + 5) \geq M - 6\dmax^2
\end{align*}
since $\dmax\geq 3$.

Combining these estimates shows that the number of possible 3-switches 
$(a_1,b_1,a_2,b_2,a_3,b_3)$ such that $L(a_1,b_1)=2$ and
$L(a_1,b_3)=-1$ is at least 
\begin{equation}
 2\left(M- 6 \dmax^2\right) 
  \geq \dfrac{2}{3}M,
\label{21forward}
\end{equation}
using the fact that $\dmax\leq \nfrac{1}{3}\, \sqrt{M}$.

Now we consider the reverse of this operation, which is
given by reversing the arrow in Figure~\ref{double-switch}.  Given 
$L'\in C(p-1,q-1)$, we need an upper bound on
the number of 6-tuples
$(a_1,b_1,a_2,b_2,a_3,b_3)$ such that $L'(a_1,b_1)=L'(a_1,b_3)=L'(a_3,b_2)=1$
and $L'(a_2,b_1)=L'(a_2,b_2)=L'(a_3,b_3)=0$.   Since the encoding
$L\in C(p,q)$ produced by this reverse operation must be
consistent with $Z$, it follows that $\{ a_2,b_1\}$ must be an
edge of $Z$.  Hence there are at most $M$ choices for $(a_2,b_1)$.
Let $\eta'_v$ denote the number of $(-1)$-defect
edges incident with $v$ in $L'$, for any vertex $v$. 
There are at most $d_{b_1} + \eta'_{b_1}$ ways to choose $a_1\in N_L(b_1)$
and at most $d_{a_1}-1+\eta'_{a_1}$ ways to choose $b_3\in N_L(a_1)\setminus
\{ b_1\}$.  
From Figure~\ref{f:possible}, if $\eta'_{a_1}=2$ then $\eta'_{b_1}=0$,
and if $\eta'_{b_1}=1$ then $\eta'_{a_1}\leq 1$.  Furthermore, $\eta'_{b_1}\leq 1$.
(Otherwise, the reverse switching would produce an encoding which is not valid.)
Therefore, the number of ways to choose $(a_1,b_3)$ with the given conditions is
at most
\[
   (d_{b_1} + \eta'_{b_1})(d_{a_1} - 1 + \eta'_{a_1})
    \leq d_{\max}\, (d_{\max} + 1)
   \leq \dfrac{4}{3}\, d_{\max}^2.
\]
Finally we
must choose $(a_3,b_2)$ such that $L'(a_3,b_2) = 1$, the vertices $a_3,b_2$
are distinct from the four vertices chosen so far and $L'(a_2,b_2) = L'(a_3,b_3)=0$.
For an upper bound, we simply ensure that $(a_3,b_2)$ is not equal to either
orientation of the two edges we have chosen so far (namely $(a_1,b_1)$ or
$(a_1,b_3)$ or their reversals).  Hence there are at most
\[
    M -4(p-1) + 2(q-1) - 4 \leq M\]
good choices for $(a_3,b_2)$. 
Therefore, the number of ways to apply the reverse operation to 
$L'\in C(p-1,q-1)$ to produce a consistent encoding $L\in C(p,q)$ is at
most $\nfrac{4}{3} \dmax^2 M^2$.

Combining this with (\ref{21forward}) shows that whenever $p+q=4$,
by double counting,
\begin{equation}
\label{B21}
 \frac{|C(p,q)|}{|C(p-1,q-1)|} \leq 2 \dmax^2 M.
\end{equation}

\bigskip

\noindent {\bf Phase 2.}\  Once Phase 1 is complete, we have reached an
encoding $L\in C(p,q)$  with $p+q\leq 3$.
If $p=0$ then Phase 2 is empty: proceed to Phase 3. Otherwise, we have
\[ (p,q) \in \{  (2,1),\, (2,0),\, (1,2),\, (1,1),\, (1,0)\}.\]
We count the number of ways to perform a 3-switch to reduce the
number of 2-defect edges by one, as shown in Figure~\ref{2-3-switch}.

\begin{figure}[ht!]
\begin{center}
\begin{tikzpicture}
\draw [-,thick,dashed] (0,0) -- (0,2);
\draw [-,thick,dashed] (2,2) -- (4,2);
\draw [-,thick,dashed] (4,0) -- (2,0);
\draw [-,line width=2.5pt] (0,2) -- (2,2);
\node [above] at (1,2.1) {$2$};
\draw [-,thick] (4,2) -- (4,0);
\draw [-,thick] (2,0) -- (0,0);
\draw [fill] (0,0) circle (0.1);
\draw [fill] (2,0) circle (0.1);
\draw [fill] (4,0) circle (0.1);
\draw [fill] (0,2) circle (0.1);
\draw [fill] (2,2) circle (0.1);
\draw [fill] (4,2) circle (0.1);
\node [below] at (0,-0.1) {$b_3$};
\node [left] at (-0.5,1) {$L$};
\node [above] at (0.0,2.1) {$a_1$};
\node [above] at (2,2.1) {$b_1$};
\node [above] at (4,2.1) {$a_2$};
\node [below] at (4,-0.1) {$b_2$};
\node [below] at (2.0,-0.1) {$a_3$};
\draw [->,line width = 1mm] (5,1) -- (6,1);
\draw [-,thick] (7,0) -- (7,2);
\draw [-,thick] (9,2) -- (11,2);
\draw [-,thick] (11,0) -- (9,0);
\draw [-,thick] (7,2) -- (9,2);
\draw [-,thick,dashed] (11,2) -- (11,0);
\draw [-,thick,dashed] (9,0) -- (7,0);
\node [below] at (7,-0.1) {$b_3$};
\node [above] at (7,2.1) {$a_1$};
\node [above] at (9,2.1) {$b_1$};
\node [right] at (11.5,1) {$L'$};
\node [above] at (11,2.1) {$a_2$};
\node [below] at (11,-0.1) {$b_2$};
\node [below] at (9,-0.1) {$a_3$};
\draw [fill] (7,0) circle (0.1);
\draw [fill] (7,2) circle (0.1);
\draw [fill] (9,2) circle (0.1);
\draw [fill] (11,2) circle (0.1);
\draw [fill] (11,0) circle (0.1);
\draw [fill] (9,0) circle (0.1);
\end{tikzpicture}
\caption{A 3-switch in Phase 2.} 
\label{2-3-switch}
\end{center}
\end{figure}
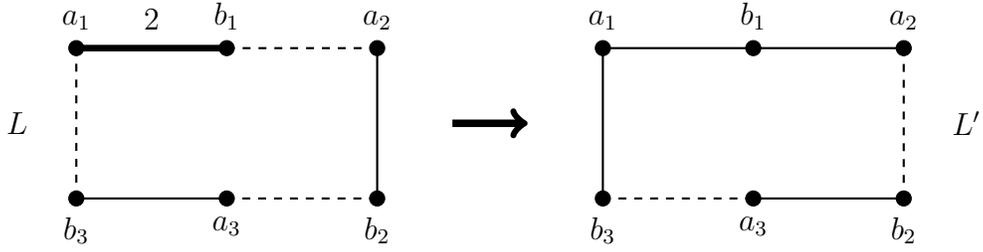
Choose an ordered pair $(a_1,b_1)$ such that $L(a_1,b_1)=2$, in 
$2p$ ways.  Next, we must choose an ordered pair $(a_2,b_2)$ such that 
$L(a_2,b_2)=1$ and $L(a_2,b_1)=0$ with $a_1,b_1,a_2,b_2$ all distinct. 
We will apply Lemma~\ref{useful-bounds}(i).  For a worst case, we make
$\eta_{b_1}$ as large as possible, and then make $\eta_{a_1}$ as large
as possible, while making all relevant values of $\zeta_v$ as small
as possible.  The worst case occurs when $(p,q)=(1,2)$ and $\eta_{b_1}=2$, which implies
that $\eta_{a_1}=0$.  This shows that there are at least
\[
 M - \Big(d_{\max}^2 + 5 \dmax - 1\Big) \geq M - 3\dmax^2
\]
good choices for $(a_2,b_2)$, using the fact that $\dmax\geq 3$.

Next, choose an ordered pair $(a_3,b_3)$ such that all six vertices
are distinct, $L(a_3,b_3)=1$ and $L(a_1,b_3)=L(a_3,b_2)=0$.  
We will apply Lemma~\ref{useful-bounds}(ii).
A worst case configuration for $(a_3,b_3)$ is shown below.
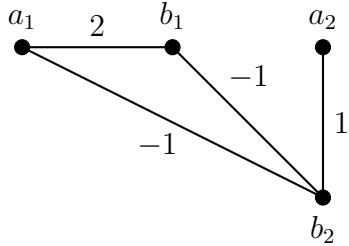
\begin{center}
\begin{tikzpicture}
\draw[fill] (0,2) circle (0.1);
\draw[fill] (2,2) circle (0.1);
\draw[fill] (4,0) circle (0.1);
\draw[fill] (4,2) circle (0.1);
\node [below] at (1.8,1.0) {$-1$};
\node [right] at (2.6,1.6) {$-1$};
\node [right] at (4.0,1.0) {$1$};
\draw[thick,-] (4,2) -- (4,0);
\draw[thick,-] (0,2) -- (4,0) -- (2,2);
\draw[thick,-] (0,2) -- (2,2);
\node [above] at (1,2) {$2$};
\node [above] at (0,2.1) {$a_1$};
\node [above] at (2,2.1) {$b_1$};
\node [above] at (4,2.1) {$a_2$};
\node [below] at (4,-0.1) {$b_2$};
\end{tikzpicture}
\end{center}
\noindent Here
\begin{align*} \eta^* = 4,\quad \zeta^* = 2, \quad & \eta_{a_1} + \eta_{b_2} = 3,\quad \zeta_{a_1} + \zeta_{b_2} = 1,\\
\sum_{x\in\hatN_L(a_1)} (\eta_x - 2\zeta_x) = 1, \qquad & \sum_{y\in\hatN_L(b_2)} (\eta_y - 2\zeta_y) = -2.
\end{align*}
Plugging these values into the bound from Lemma~\ref{useful-bounds}(ii), the number of good
choices for $(a_3,b_3)$ is at least
\begin{align*}
 M - \biggl(2d_{\max}^2  + 9\dmax - 1\biggr)
   \geq M - 5 d_{\max}^2.
\end{align*}
Combining these expressions, we conclude that there are at least
\begin{equation}
 2\left(M - 3\dmax^2\right)\left(M - 5\dmax^2\right) \geq \dfrac{16}{27} M^2
\label{forward2}
\end{equation}
valid choices for $(a_1,b_1,a_2,b_2,a_3,b_3)$, using the stated
upper bound on $\dmax$.

For the reverse operation, let $L'\in C(p-1,q)$ where
$p\geq 1$ and $p+q\leq 3$.  We need an upper bound
on the number of 6-tuples $(a_1,b_1,a_2,b_2,a_3,b_3)$ with
$L(a_1,b_1)=L(a_1,b_3) = L(a_2,b_1) = L(a_3,b_2)=1$ and
$L(a_2,b_2)=L(a_3,b_3)=0$.  There are at most 
\[ M-4(p-1)+2q\leq M+4\]
choices for $(a_1,b_1)$ with $L(a_1,b_1)=1$, and then there are
at most 
\[ (d_{a_1} - 1 + \eta'_{a_1})(d_{b_1} - 1 + \eta'_{b_1}) \leq \dmax^2\]
 choices for $(a_2,b_3)$, where (as in Phase 1), $\eta'_v$ is the number of $(-1)$-defect
edges incident with $v$ in $L'$.  This uses the fact that there 
are at most two defect edges in $L'$, and hence 
$\eta'_{a_1} + \eta'_{b_1}\leq 2$, by choice of $(a_1,b_1)$.
Finally there are at most 
\[ M - 4(p-1) + 2q - 6 \leq M-2\] 
valid choices for $(a_3,b_2)$, where for an upper bound we just
avoid choosing any orientation of the three edges (namely $(a_1,b_1)$,  
$(a_2,b_1)$, $(a_1,b_3)$) which have already been chosen.
Hence the number of
6-tuples where the reverse operation can be performed is at most
\[ \dmax^2 (M+4)(M-2) \leq \dfrac{83}{81}\, M^2,\]
since $M\geq 9\dmax^2\geq 81$.

Combining this with (\ref{forward2}), it follows that 
for $(p,q)\in \{ (2,1),\, (2,0),\, (1,2),\, (1,1),\, (1,0)\}$, 
we have
\begin{equation} 
\frac{|C(p,q)|}{|C(p-1,q)|} \leq \dfrac{83}{48}\, \dmax^2 < 2\dmax^2.
\label{B2}
\end{equation}
 
\bigskip

\noindent {\bf Phase 3.}\ After Phase 2, we may
suppose that $p=0$.  Let $L\in C(0,q)$ where $q\in \{ 1,2,3\}$.  
We count the number of 6-tuples $(a_1,b_1,a_2,b_2,a_3,b_3)$
where a 3-switch can be performed with $L(a_2,b_1)=-1$.
Performing this 3-switch
will produce $L'\in C(0,q-1)$, as illustrated in Figure~\ref{1switch}.

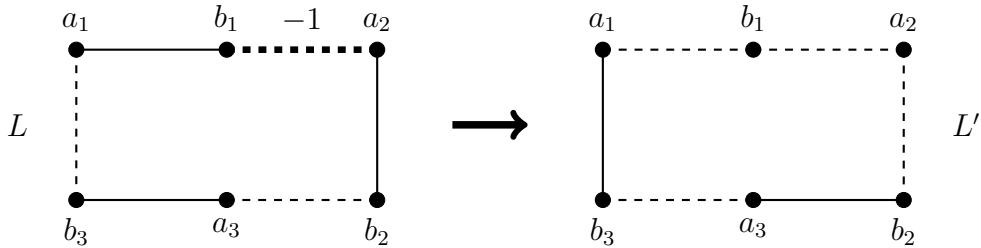
\begin{figure}[ht!]
\begin{center}
\begin{tikzpicture}
\draw [-,thick,dashed] (0,0) -- (0,2);
\draw [-,line width=2.5pt,dashed] (2,2) -- (4,2);
\node [above] at (3,2.1) {$-1$};
\draw [-,thick,dashed] (4,0) -- (2,0);
\draw [-,thick] (0,2) -- (2,2);
\draw [-,thick] (4,2) -- (4,0);
\draw [-,thick] (2,0) -- (0,0);
\draw [fill] (0,0) circle (0.1);
\draw [fill] (2,0) circle (0.1);
\draw [fill] (4,0) circle (0.1);
\draw [fill] (0,2) circle (0.1);
\draw [fill] (2,2) circle (0.1);
\draw [fill] (4,2) circle (0.1);
\node [below] at (0,-0.1) {$b_3$};
\node [left] at (-0.5,1) {$L$};
\node [above] at (0.0,2.1) {$a_1$};
\node [above] at (2,2.1) {$b_1$};
\node [above] at (4,2.1) {$a_2$};
\node [below] at (4,-0.1) {$b_2$};
\node [below] at (2.0,-0.1) {$a_3$};
\draw [->,line width = 1mm] (5,1) -- (6,1);
\draw [-,thick] (7,0) -- (7,2);
\draw [-,thick,dashed] (9,2) -- (11,2);
\draw [-,thick] (11,0) -- (9,0);
\draw [-,thick,dashed] (7,2) -- (9,2);
\draw [-,thick,dashed] (11,2) -- (11,0);
\draw [-,thick,dashed] (9,0) -- (7,0);
\node [below] at (7,-0.1) {$b_3$};
\node [above] at (7,2.1) {$a_1$};
\node [above] at (9,2.1) {$b_1$};
\node [above] at (11,2.1) {$a_2$};
\node [right] at (11.5,1) {$L'$};
\node [below] at (11,-0.1) {$b_2$};
\node [below] at (9,-0.1) {$a_3$};
\draw [fill] (7,0) circle (0.1);
\draw [fill] (7,2) circle (0.1);
\draw [fill] (9,2) circle (0.1);
\draw [fill] (11,2) circle (0.1);
\draw [fill] (11,0) circle (0.1);
\draw [fill] (9,0) circle (0.1);
\end{tikzpicture}
\caption{A 3-switch in Phase 3.}
\label{1switch}
\end{center}
\end{figure}
There are $2q$ ways to choose $(b_1,a_2)$, and
at least $d_{b_1}+1$ ways to choose $a_1\in N_L(b_1)$.
Then there are at least $d_{a_2}$ ways to choose $b_2\in N_L(a_2)\setminus \{ a_1\}$.
(Note that the presence of other $(-1)$-defect edges incident with
$b_1$ or $a_2$ only helps here.)  Finally, we must choose $(a_3,b_3)$
with $L(a_3,b_3)=1$
such that all vertices are distinct, $L(a_3,b_2)=0$ and $L(a_1,b_3)=0$. 
Again, an upper bound on the number of bad choices for $(a_3,b_3)$ is given by
Lemma~\ref{useful-bounds}(ii), noting that now $\zeta_v = 0$ for all vertices $v$.
The worst case is attained with $q=3$, for example when the
defect edges are as shown in Figure~\ref{worstcase}, but with the
edge label on $\{ a_1,b_1\}$ changed from 2 to 1.
Here
\[ \eta^* = 6,\quad \eta_{a_1} + \eta_{b_2} = 3,  \quad
\sum_{x\in\hatN_L(a_1)}\eta_x = 5,\quad
 \sum_{y\in \hatN_L(b_2)}\eta_y = 4.
\]
Substituting these values into Lemma~\ref{useful-bounds}(ii), the number
of valid choices for $(a_3,b_3)$ is at least
\begin{align*}
  M - \left(2d_{\max}^2 + 10 d_{\max} + 9\right)
  \geq M - 7 d_{\max}^2.
\end{align*}
Hence the number of 3-switches which can be performed in $L$
to reduce the number of ($-1$)-defect edges by exactly one
is at least
\begin{align}
\label{1forward} 
2q(d_{b_1}+1)\, d_{a_2}\, (M - 7 d_{\max}^2) \geq 4q(M - 7 d_{\max}^2)  \geq \dfrac{8}{9} M,
\end{align}
using the given bounds on $\dmin$ and $\dmax$.

For the reverse operation, let $L'\in C(0,q-1)$, where
$q\in \{ 1,2,3\}$.  We need an upper bound on the number of
6-tuples such that $L(a_1,b_3)=L(a_3,b_2)=1$, $L(a_1,b_1)=L(b_1,a_2)=
L(a_2,b_2)=L(a_3,b_3) = 0$ and $\{ a_2,b_1\}$ is an edge of $Z$.
There are at most $M$ choices for $(a_2,b_1)$ satisfying
the latter condition, then at most
$M + 2(q-1) - 2(d_{a_2} + d_{b_1})\leq M$ ways to choose 
$(a_3,b_2)$ with $L(a_3,b_2)=1$ 
and $a_1,a_3,b_2,b_3$ all distinct. 
(Observe that the presence of any additional $(-1)$-defect
edges incident with $a_2$ or $b_1$ can only help here.)
Similarly, there at most
$M$ ways to choose $(a_1,b_3)$.
Hence the number of reverse operations is at most $M^3$. 

Combining this with (\ref{1forward}) shows that
\begin{equation}
 \frac{|C(0,q)|}{|C(0,q-1)|} \leq \dfrac{9}{8} M^2
\label{B1}
\end{equation}
for $q\in \{ 1,2,3\}$, by double counting.

\medskip

\noindent {\bf Consolidation.}\ 
Define
\[ a = 2 \dmax^2 M, \qquad 
   b = 2 \dmax^2, \qquad
      c = \dfrac{9}{8} M^2. \]
It follows from (\ref{B21}), (\ref{B2}) and (\ref{B1}) that
\begin{align*}
 \frac{|\mathcal{L}^\ast(Z)|}{|\Omegau|} &=
  \sum_{p=0}^2\,\sum_{q=0}^3\, \frac{|C(p,q)|}{|C(0,0)|} \\
   &\leq
  1 + b + b^2 + c + bc + b^2 c + abc
    + c^2  + bc^2  + a c^2  + c^3  \\
 &\leq 2 M^6,
\end{align*}
using the upper bound on $\dmax$ and the fact that $M\geq 9\dmax^2 \geq 81$.
This completes the proof of Lemma~\ref{3switches}.
\end{proof}

Since $M\leq \dmax n$, the bound $2 M^6$ is at most
a factor $n$ bigger than the analogous bound $2d^6 n^5$ given
in~\cite[Lemma 4]{CDG} in the regular case.

We can now quickly complete the proof of Theorem~\ref{main}.

\begin{proof}[Proof of Theorem~\ref{main}]
Recall the definitions from Section~\ref{ss:flow}.
We wish to apply (\ref{flowbound}).
It follows from the configuration model (see~\cite[Equation (1)]{McKW91}) that
the set $\Omegau$ has size
\begin{equation} |\Omegau| \leq \frac{M!}{2^{M/2}\, (M/2)!\, \prod_{j=1}^n d_j!} \leq
          \exp\left( \nfrac{1}{2} \, M\log(M) \right).
\label{size}
\end{equation}
Hence the smallest stationary probability $\pi^\ast$ satisfies
$\log(1/\pi^\ast) = \log(|\Omegau|) \leq \nfrac{1}{2} M\log(M)$. 
Next, $\ell(f)\leq M/2$ since each transition along a canonical path
replaces an edge of $G$ by an edge of $G'$.

Finally, if $e=(Z,Z')$ is a transition of the switch chain then
$1/Q(e) = 6\, a(\dvec) \leq M^2$.
Combining this with Lemma~\ref{fbound} gives
$\rho(f) \leq 2 \dmax^{14}\, M^8$.
Substituting these expressions into (\ref{flowbound}) gives
the claimed bound on the mixing time.
\end{proof}

\section{The directed switch chain}\label{s:directed}

A directed graph (digraph) $G=(V,A)$ consists of a finite set of vertices
and a set $A$ (or $A(G)$) of arcs,
where each arc is an ordered pair of distinct vertices.
We take $V=[n]$ for some positive integer $n$.

Recall that $\Omegad$ is the set of all directed graphs 
with directed degree sequence $\vecdvec$, as defined in Section~\ref{s:intro}.
The directed switch Markov chain, denoted $\Md$, has state space
$\Omegad$ and transitions described by the following procedure:
from the current digraph $G\in\Omegad$, choose an unordered pair
$\{ (i,j),(k,\ell)\}$ of distinct arcs of $G$ uniformly at random. 
If $i,j,k,\ell$ are distinct and $\{ (i,j),(k,\ell)\} \cap A(G) = \emptyset$
then delete the arcs $(i,j)$, $(k,\ell)$ from $G$ and add the arcs $(i,\ell)$, $(k,j)$ 
to obtain the new state; otherwise, remain at $G$.
If distinct digraphs $G, G'\in\Omegad$ are related by a directed switch then 
$P(G,G') = 1/\binom{m}{2} = P(G',G)$.
Hence the directed switch chain $\Md$ is symmetric, so the stationary distribution is 
uniform over $\Omegad$. The chain is also aperiodic, since for any $G\in\Omegad$ 
there are at least $m$ pairs of incident edges (for example, pairs of the form 
$\{( i,j),\, (j,\ell)\}$). This
implies that $P(G,G) \geq m/\binom{m}{2} > 0$.
As discussed in Section~\ref{s:intro}, unlike for undirected graphs, the
directed switch chain is not irreducible on $\Omegad$ for all directed
degree sequences $\vecdvec$. Instead, we will assume throughout this section
that $\vecdvec$ is a switch-irreducible degree sequence; 
that is, we assume that $\Md$ is irreducible.

In~\cite{directed}, a multicommodity flow analysis was given
for $\Md$ for the case of regular directed degree sequences.
We now show how to adapt this proof to handle irregular directed
degree sequences which satisfy the conditions of Theorem~\ref{main-directed}.
One result used to define the multicommodity flow, namely~\cite[Lemma 2.3]{directed},
must be reproved here without using the regularity assumption.
Having done that, we may use exactly the same multicommodity flow
as defined in~\cite{directed}. This is discussed in Section~\ref{s:defining-directed}
below.
Then, the flow is analysed in Section~\ref{s:analysing-directed}.
As in the undirected case (Section~\ref{s:undirected}),
we must reprove a critical counting lemma bounding the number of encodings,
without using regularity.

\subsection{Defining the flow}\label{s:defining-directed}

The overall structure of the multicommodity flow argument defined in~\cite{directed}
is very similar to the undirected case (on which it was based).  
Again, given two digraphs $G,G'\in\Omegad$ we consider the symmetric
difference $G\triangle G'$ as a 2-arc-coloured digraph.
A pairing now consists of a bijection from the blue in-arcs at $v$ to the red in-arcs
at $v$, and a bijection from the blue out-arcs at $v$ to the red out-arcs at $v$,
for each vertex $v$.  With respect to a fixed vertex, the symmetric difference
can be decomposed into a sequence of \emph{circuits}, and circuits are then
decomposed into 1-circuits or 2-circuits.  Both the colour and the direction of the
arcs alternate around each 1-circuit or 2-circuit.  Then each 1-circuit 
and 2-circuit must be processed, in order, using a sequence of switches which
form part of the canonical path from $G$ to $G'$.

Some extra cases arise which do not
occur in the undirected case, including the special case that the 2-circuit has
precisely 6 edges and 3 vertices, which we call 
a \emph{triangle}: see Figure~\ref{fig:triangle}.
\begin{figure}[ht!]
\begin{center}
\begin{tikzpicture}[scale=1.5]
\draw[fill] (1,0) circle (0.08);
\draw[fill] (2,1.7) circle (0.08);
\draw[fill] (3,0) circle (0.08);
\draw[thick, ->] (2.85, -0.05) -- (1.15, -0.05); 
\draw[thick, <-] (1.85, 1.6) -- (1.05, 0.15);
\draw[thick, ->] (2.05, 1.55) -- (2.82, 0.15);
\draw[thick, dashed, <-] (2.15, 1.6) -- (2.95, 0.15); 
\draw[thick, dashed, <-] (1.15, 0.1) -- (1.95, 1.55);
\draw[thick, dashed, ->] (1.15, 0.05) -- (2.80, 0.05);
\node [left] at (0.9, 0) {$v_0$};
\node [left] at (1.9, 1.8) {$v_1$};
\node [right] at (3.1,0) {$v_2$};
\end{tikzpicture}
\end{center}
\caption{A triangle in the symmetric difference of $G$ and $G'$}
\label{fig:triangle}
\end{figure}
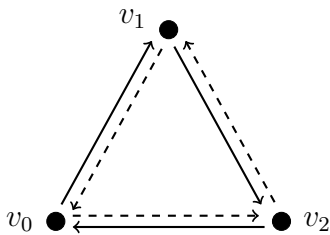
In order to process a triangle in the regular case, 
we use some results of LaMar~\cite{lamar,lamar2} about the
structure of directed graphs for which the directed switch chain 
is irreducible.  
Let $(x,U)$ denote
the set of all arcs of the form $(x,u)$ with $u\in U$,
and similarly for $(U,x)$.
Following LaMar, define the four vertex sets
\begin{equation} \label{classes}
\begin{split}
& U^0 = \{x \in [n] \setminus U: (x,U) \cup (U,x) \subset A(G)^c \}, \\
& U^- = \{x \in [n] \setminus U: (x,U) \subset A(G) \text{ and } (U,x) \subset A(G)^c \},\\
& U^+ = \{x \in [n] \setminus U: (x,U) \subset A(G)^c \text{ and } (U,x) \subset A(G) \},\\
& U^{\pm} = \{x \in [n] \setminus U: (x,U) \cup (U,x) \subset A(G) \}
\end{split}
\end{equation}
where $A(G)^c$ denotes the set of all non-arcs of $G$. 

Combining LaMar's results~\cite[Theorems 3.3 and 3.4]{lamar} gives
useful characterisation of degree sequences for which
$\Md$ is not irreducible which we restate here, for convenience.

\begin{lemma} \emph{\cite{lamar}}\
The directed switch chain $\Md$ fails to be irreducible if and only
if for every $G\in\Omegad$ there exists a vertex subset $U=\{ v_0,v_1,v_2\}$
such that $G[U]$ is a directed 3-cycle and the four sets $U^0$, $U^-$, $U^+$,
$U^{\pm}$ form a partition of $[n]\setminus U$, such that in addition,
\begin{itemize}
\item[\emph{(U1)}] no arcs from $U^0\cup U^+$ to $U^0\cup U^-$ are present, and
\item[\emph{(U1)}] all (non-loop) arcs from $U^-\cup U^{\pm}$ to $U^+\cup U^{\pm}$ are present.
\end{itemize}
\label{characterisation}
\end{lemma}

As in~\cite{directed}, we say that a vertex $x\not\in U$ is a \emph{useful neighbour} 
for the directed 3-cycle on $U$ if $u\not\in U^0 \cup U^-\cup U^+ \cup U^{\pm}$.
Similarly, we say that the arc $(x,y)$ is a \emph{useful arc} for the directed 3-cycle
on $U$ if one of the following conditions holds:
\begin{itemize}
\item[(i)] $(x,y)\in A(G)$, with $x\in U^0\cup U^+$ and $y\in U^0\cup U^-$, or
\item[(ii)] $(x,y)\not\in A(G)$, with $x\in U^-\cup U^{\pm}$ and $y\in U^+\cup U^{\pm}$.
\end{itemize}
The following lemma 
is needed in the definition of the multicommodity flow, in order
to handle 2-circuits which are triangles.
The proof given in the regular case (\cite[Lemma 2.3]{directed})
used the regularity assumption together with LaMar's characterisation 
(Lemma~\ref{characterisation}). 
Below we provide a more direct proof which does not rely on regularity.

\begin{lemma}
\label{triangle}
Suppose that $\vecdvec$ is a switch-irreducible directed degree sequence,
and that $G \in \Omegad$ contains a set of three vertices $U=\{v_0,v_1,v_2\}$ 
such that the induced digraph $G[U]$ is a directed 3-cycle. Then there exists a 
useful neighbour or a useful arc for this 3-cycle.
\end{lemma}

\begin{proof}
We apply LaMar's characterisation, stated above as Lemma~\ref{characterisation}.
Since $\vecdvec$ is switch-irreducible, Lemma~\ref{characterisation} guarantees
that either there exists a vertex 
$x\in [n]\setminus U$ which does not belong to $U^0\cup U^-\cup U^+\cup U^{\pm}$,
or there exists an arc $(x,y)$ which contradicts either (U1) or (U2).
In the first case $x$ is a useful neighbour of the 3-cycle on $U$, while
in the second case $(x,y)$ is a useful arc for the 3-cycle on $U$.
\end{proof}

With this lemma in hand, the same definition of multicommodity flow
from~\cite{directed} may be used, for any directed degree sequence
$\vecdvec$ which satisfies the
conditions of Theorem~\ref{main-directed}. 

\subsection{Analysing the flow}\label{s:analysing-directed}

Given $G,G',Z \in \Omegad$, we define the $n \times n$ matrix $L$ to be 
an \textsl{encoding} of $Z$ (with respect to $G,G'$) by setting $L + Z = G+G'$, 
as for the undirected case. Treating $L$ as an arc-labelled digraph, we label 
the arcs with $-1,1$ or $2$ (zero entries correspond to arcs which are
absent). A $(-1)$-\textsl{defect arc} is an arc labelled $-1$.
This is an arc which is absent in both $G$ and $G'$ but present in $Z$.
Similarly, a $2$-\emph{defect arc} is an arc labelled $2$. 
This is an arc which is present in both $G$ and $G'$ but absent in $Z$.
We write $L(a,b)$ for the label of the directed arc $(a,b)$ in the 
encoding $L$. 

The next lemma is the directed analogue of Lemma~\ref{oldstuff}, collecting
together some important results from~\cite{directed}: the proofs
given there did not rely on regularity, and so they extend without change
to irregular directed degree sequences.

\begin{lemma} \label{configuration}
Given $G,G' \in \Omegad$ with symmetric difference $G \triangle G'$, 
let $(Z,Z')$ be a transition on the canonical path from $G$ to $G'$ with 
respect to the pairing $\psi \in \Psi(G,G')$. Let $L$ be the encoding of $Z$ with respect to $(G,G')$. Then the following statements hold:
\begin{itemize}
\item[\emph{(i)}] 
\emph{(\cite[Lemma 5.2]{directed})}\
Given $(Z,Z'), L$ and $\psi$, there are at most four possibilities for $(G,G')$ such that $(Z,Z')$ is a transition along the canonical path from $G$ to $G'$ corresponding to $\psi$ and $L$ is an encoding for $Z$ with respect to $(G,G')$.
\item[\emph{(ii)}]
\emph{(\cite[Lemma 5.1]{directed})}\
There are at most five defect arcs in $L$.
The digraph consisting of the defect arcs in $L$ must form a subdigraph of 
one of the possible labelled digraphs shown in Figure~\ref{f:configuration},
up to the symmetries described below.
\end{itemize}
\end{lemma}

Define the \emph{arc-reversal} operator $\zeta$,
which acts on a digraph
$G$ by reversing every arc in $G$; that is, replacing $(u,v)$ by $(v,u)$
for every arc $(u,v)\in A(G)$.  
In Figure~\ref{f:configuration}, $\{ \mu,\nu\} = \{ -1,2\}$ and $\{ \xi,\omega\} = \{-1,2\}$ independently, giving four symmetries obtained by exchanging these pairs. 
We can also apply the operation $\zeta$ to reverse the orientation of all arcs. 
Hence each digraph shown in Figure~\ref{f:configuration} represents up to eight 
possible digraphs. 

\begin{figure}[ht!] 
\begin{center}
\begin{tikzpicture}[scale = 1.2]
\draw[fill] (1,2) circle (0.1);
\draw[fill] (1,1) circle (0.1);
\draw[fill] (1,3) circle (0.1);
\draw[fill] (0,2) circle (0.1);
\draw[fill] (2,2) circle (0.1);
\draw[fill] (2,1) circle (0.1);
\draw[thick,->] (0.85,2) -- (0.15,2);
\draw[thick,->] (1,2.15) -- (1,2.85);
\draw[thick,->] (1,1.85) -- (1,1.15);
\draw[thick,->] (1.85,2) -- (1.15,2);
\draw[thick,->] (2,1.85) -- (2,1.15);
\node [right] at (1,2.5) {$\mu$};
\node [above] at (0.5,2) {$\mu$};
\node [left] at (1,1.5) {$\nu$};
\node [above] at (1.5,2) {$\omega$};
\node [right] at (2,1.5) {$\xi$};

\draw[fill] (3,2) circle (0.1);
\draw[fill] (4,2) circle (0.1);
\draw[fill] (5,2) circle (0.1);
\draw[fill] (4,1) circle (0.1);
\draw[fill] (4,3) circle (0.1);
\draw[thick,->] (3.85,2) -- (3.15,2);
\draw[thick,->] (4,2.15) -- (4,2.85);
\draw[thick,->] (4,1.85) -- (4,1.15);
\draw[thick,->] (4.85,2) -- (4.15,2);
\draw[thick,->] (4.9,1.9) -- (4.1,1.1);
\node [right] at (4,2.5) {$\mu$};
\node [above] at (3.5,2) {$\mu$};
\node [left] at (4,1.5) {$\nu$};
\node [above] at (4.5,2) {$\omega$};
\node [right] at (4.5,1.5) {$\xi$};

\draw[fill] (6,2) circle (0.1);
\draw[fill] (7,2) circle (0.1);
\draw[fill] (8,2) circle (0.1);
\draw[fill] (7,1) circle (0.1);
\draw[fill] (7,3) circle (0.1);
\draw[thick,->] (6.85,2) -- (6.15,2);
\draw[thick,->] (7,2.15) -- (7,2.85);
\draw[thick,->] (7,1.85) -- (7,1.15);
\draw[thick,->] (7.85,2) -- (7.15,2);
\draw[thick,->] (7.9,1.9) -- (7.1,1.1);
\node [right] at (7,2.5) {$\mu$};
\node [above] at (6.5,2) {$\nu$};
\node [left] at (7,1.5) {$\mu$};
\node [above] at (7.5,2) {$\omega$};
\node [right] at (7.5,1.5) {$\xi$};

\draw[fill] (9,2) circle (0.1);
\draw[fill] (10,2) circle (0.1);
\draw[fill] (9,1) circle (0.1);
\draw[fill] (10,1) circle (0.1);
\draw[fill] (10,3) circle (0.1);
\draw[thick,->] (9.85,2) -- (9.15,2);
\draw[thick,->] (10,2.15) -- (10,2.85);
\draw[thick,->] (9.95,1.85) -- (9.95,1.15);
\draw[thick,->] (10.05,1.15) -- (10.05,1.85);
\draw[thick,->] (9.85,1) -- (9.15,1);
\node [right] at (10,2.5) {$\mu$};
\node [above] at (9.5,2) {$\mu$};
\node [left] at (10,1.5) {$\nu$};
\node [right] at (10,1.5) {$\omega$};
\node [below] at (9.5,1) {$\xi$};

\draw[fill] (1.5,-1) circle (0.1);
\draw[fill] (1.5,-2) circle (0.1);
\draw[fill] (1.5,-3) circle (0.1);
\draw[fill] (0.5,-2) circle (0.1);
\draw[fill] (0.5,-3) circle (0.1);
\draw[thick,->] (1.35,-2) -- (0.65,-2);
\draw[thick,->] (1.5,-1.85) -- (1.5,-1.15);
\draw[thick,->] (1.45,-2.15) -- (1.45,-2.85);
\draw[thick,->] (1.55,-2.85) -- (1.55,-2.15);
\draw[thick,->] (1.35,-3) -- (0.65,-3);
\node [right] at (1.5,-1.5) {$\mu$};
\node [above] at (1,-2) {$\nu$};
\node [left] at (1.5,-2.5) {$\mu$};
\node [right] at (1.5,-2.5) {$\omega$};
\node [below] at (1,-3) {$\xi$};

\draw[fill] (3.5,-2) circle (0.1);
\draw[fill] (4.5,-1) circle (0.1);
\draw[fill] (4.5,-2) circle (0.1);
\draw[fill] (4.5,-3) circle (0.1);
\draw[thick,->] (4.35,-2) -- (3.65,-2);
\draw[thick,->] (4.5,-1.85) -- (4.5,-1.15);
\draw[thick,->] (4.45,-2.15) -- (4.45,-2.85);
\draw[thick,->] (4.55,-2.85) -- (4.55,-2.15);
\draw[thick,->] (4.4,-2.9) -- (3.6,-2.1);
\node [right] at (4.5,-1.5) {$\mu$};
\node [above] at (4,-2) {$\mu$};
\node [left] at (4.5,-2.5) {$\nu$};
\node [right] at (4.5,-2.5) {$\omega$};
\node [left] at (4,-2.6) {$\xi$};

\draw[fill] (6.5,-2) circle (0.1);
\draw[fill] (7.5,-1) circle (0.1);
\draw[fill] (7.5,-2) circle (0.1);
\draw[fill] (7.5,-3) circle (0.1);
\draw[thick,->] (7.35,-2) -- (6.65,-2);
\draw[thick,->] (7.5,-1.85) -- (7.5,-1.15);
\draw[thick,->] (7.45,-2.15) -- (7.45,-2.85);
\draw[thick,->] (7.55,-2.85) -- (7.55,-2.15);
\draw[thick,->] (7.4,-2.9) -- (6.6,-2.1);
\node [right] at (7.5,-1.5) {$\mu$};
\node [above] at (7,-2) {$\nu$};
\node [left] at (7.5,-2.5) {$\mu$};
\node [right] at (7.5,-2.5) {$\omega$};
\node [left] at (7,-2.6) {$\xi$};

\draw[fill] (9,-2) circle (0.1);
\draw[fill] (10,-1) circle (0.1);
\draw[fill] (10,-2) circle (0.1);
\draw[fill] (10,-3) circle (0.1);
\draw[thick,->] (9.85,-2) -- (9.15,-2);
\draw[thick,->] (10,-1.85) -- (10,-1.15);
\draw[thick,->] (9.95,-2.15) -- (9.95,-2.85);
\draw[thick,->] (10.05,-2.85) -- (10.05,-2.15);
\draw[thick,->] (9.9,-2.9) -- (9.1,-2.1);
\node [right] at (10,-1.5) {$\nu$};
\node [above] at (9.5,-2) {$\mu$};
\node [left] at (10,-2.5) {$\mu$};
\node [right] at (10,-2.5) {$\omega$};
\node [left] at (9.5,-2.6) {$\xi$};

\draw[-] (2.5,-3.1) -- (2.5,3.1);
\draw[-] (5.5,-3.1) -- (5.5,3.1);
\draw[-] (8.5,-3.1) -- (8.5,3.1);
\draw[-] (-0.1,0) -- (10.6,0);

\end{tikzpicture}
\caption{Possible configurations of defect arcs, up to symmetries.}
\label{f:configuration}
\end{center}
\end{figure}
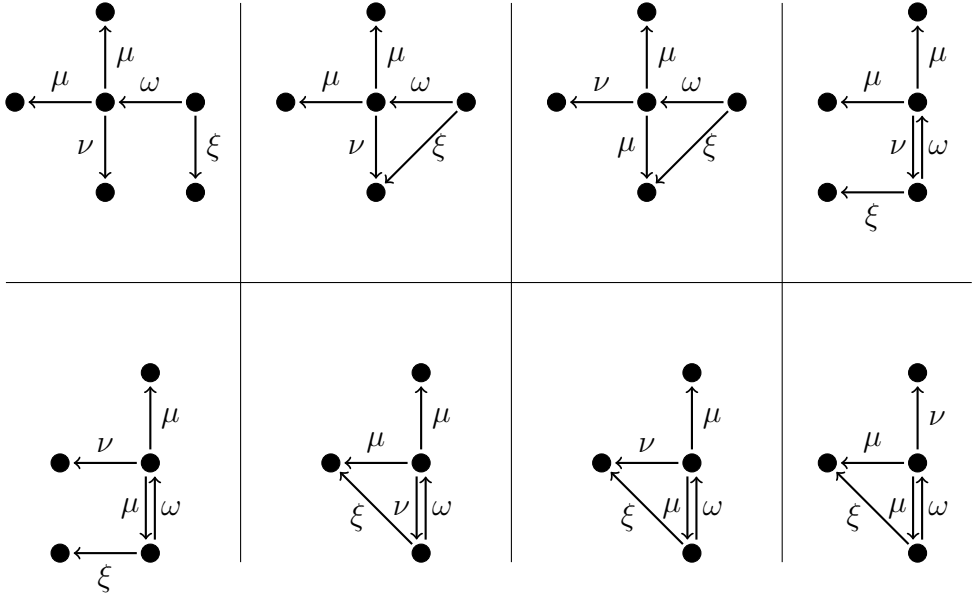

We now generalize the notion of an encoding: given a directed
degree sequence $\vecdvec$, an encoding is
any $n \times n$ matrix $L$ with 
entries in $\{ -1,0,1,2 \}$ such that the $j$th row sum is the out-degree 
$d_j^+$ and 
the $j$th column sum is the in-degree $d_j^-$, for all $j\in [n]$.
We say that an encoding $L$ is \textsl{consistent} with $Z$ if every entry of 
$L+Z$ belongs to $\{0,1,2\}$, and we say that an encoding $L$ is 
\textsl{valid} if $L$ satisfies the conclusion of Lemma~\ref{configuration}(ii). 
Let $\mathcal{L}(Z)$ be the set of valid encodings that are consistent with $Z$.

The next result is proved just as in the regular case, 
see~\cite[Lemma 5.7]{directed}, since the regularity assumption
was not used in the proof given there.

\begin{lemma}  \label{effee}
\emph{(\cite[Lemma 5.7]{directed})}\
The load $f(e)$ on the transition $e=(Z,Z')$ satisfies  
$$ f(e) \leq 4 \, \sdmax^{16} \,\frac{|\mathcal{L}(Z)|}{|\Omegad|^2}.$$
\end{lemma}

As in Section~\ref{s:undirected}, we can extend the directed
switch operation to encodings,
ensuring that we never create a label outside the set $\{ -1,0,1,2\}$.
We wish to use switchings on encodings to prove an upper bound on the ratio
$|\mathcal{L}(Z)|/|\Omegad|$, so we can substitute this bound back into 
Lemma~\ref{effee}.
In the regular case, this was achieved using 
Lemma~\cite[Lemma 5.5]{directed},
which proved that from any encoding in $\mathcal{L}(Z)$, one could obtain
a digraph in $\Omegad$ using at most three switches. 
But the proof of this ``critical lemma'' relied heavily on regularity, 
so we need a new approach here.

As in the undirected case (Section~\ref{s:undirected}), we introduce
a less tightly constrained operation
for removing defect edges.
A \textsl{directed $3$-switch} is described by a 6-tuple 
$(a_1, b_1, a_2, b_2, a_3, b_3)$ of distinct vertices of $G$ such that the 
arcs $(a_1,b_1)$, $(a_2, b_2) $, $(a_3, b_3)$ are all present in $G$ 
and the arcs $(a_2, b_1)$, $(a_3, b_2)$, $(a_1, b_3 )$ are not. 
The directed 3-switch deletes the three arcs $(a_1,b_1)$, $(a_2, b_2 )$, 
$(a_3, b_3) $ from the arc set, and replaces them with $(a_2, b_1)$, 
$(a_3, b_2) $, $(a_1, b_3)$ as shown in Figure~\ref{f:directed3switch}.

\begin{figure}[ht!] 
\begin{center}
\begin{tikzpicture}[scale = 1]
\draw[fill] (0,0) circle (0.1);
\draw[fill] (0,2) circle (0.1);
\draw[fill] (2,0) circle (0.1);
\draw[fill] (2,2) circle (0.1);
\draw[fill] (4,0) circle (0.1);
\draw[fill] (4,2) circle (0.1);
\draw[fill] (8,0) circle (0.1);
\draw[fill] (8,2) circle (0.1);
\draw[fill] (10,0) circle (0.1);
\draw[fill] (10,2) circle (0.1);
\draw[fill] (12,0) circle (0.1);
\draw[fill] (12,2) circle (0.1);
\draw[thick,->] (0.15,2) -- (1.85,2);
\draw[thick,->] (4,1.85) -- (4,0.15);
\draw[thick,->] (1.85,0) -- (0.15,0);
\draw[dashed,thick,->] (3.85,2) -- (2.15,2);
\draw[dashed,thick,->] (2.15,0) -- (3.85,0);
\draw[dashed,thick,->] (0,1.85) -- (0,0.15);
\node [below] at (0,-0.1) {$b_3$};
\node [below] at (2,-0.1) {$a_3$};
\node [below] at (4,-0.1) {$b_2$};
\node [above] at (0,2.1) {$a_1$};
\node [above] at (2,2.1) {$b_1$};
\node [above] at (4,2.1) {$a_2$};
\draw[dashed,thick,->] (8.15,2) -- (9.85,2);
\draw[dashed,thick,->] (12,1.85) -- (12,0.15);
\draw[dashed,thick,->] (9.85,0) -- (8.15,0);
\draw[thick,->] (11.85,2) -- (10.15,2);
\draw[thick,->] (10.15,0) -- (11.85,0);
\draw[thick,->] (8,1.85) -- (8,0.15);
\node [below] at (8,-0.1) {$b_3$};
\node [below] at (10,-0.1) {$a_3$};
\node [below] at (12,-0.1) {$b_2$};
\node [above] at (8,2.1) {$a_1$};
\node [above] at (10,2.1) {$b_1$};
\node [above] at (12,2.1) {$a_2$};
\draw [->,line width = 1mm] (5.5,1) -- (6.5,1);
\end{tikzpicture}
\caption{A directed 3-switch.}
\label{f:directed3switch}
\end{center}
\end{figure}
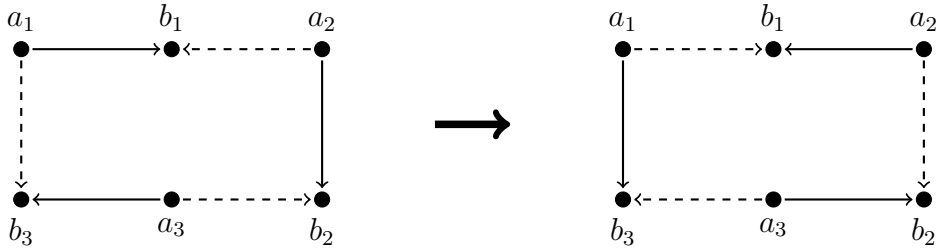
The directed 3-switch can also be extended to encodings, making sure
that after the 3-switch, all labels belong to $\{ -1,0,1,2\}$.

Let $C(p,q)$ be the set of encodings in $\mathcal{L}(Z)$ with precisely $p$ defect 
arcs  labelled 2
and precisely $q$ defect arcs labelled $-1$, for $p,q \in \{0,1,2,3\}$. 
Then $\Omegad = C(0,0)$ and 
$$ \mathcal{L}(Z) =  \bigcup_{p=0}^3\,\bigcup_{q=0}^3\,C(p,q),$$
where this union is disjoint. 
(Note that $C(3,3) = \emptyset$ as there are at most five defect arcs,
by Lemma~\ref{configuration}(ii).)

In Section~\ref{s:undirected} we used a special switch when the maximum number
of defect edges were present (that is, in Phase 1): 
this switch reduced both the number of 2-defect
edges and $(-1)$-defect edges by one. 
We were able to do this by proving extra structural information about the
defect edges (see Lemma~\ref{structure}).  Unfortunately, we were unable to
prove the analogous result in the directed case. (The main difficulty
arises from 2-circuits which are triangles.) 
Hence we will proceed by removing one defect per directed 3-switch,
requiring at most 5 directed 3-switches to transform an encoding
in $\mathcal{L}(Z)$ into an element of $\Omegad$. 
Since Phase 1 is missing from our analysis in the directed setting, 
we rename Phase 2 as Phase A and Phase 3 as Phase B.  

If $L \in C(p,q)$ then there are precisely $m - 2p + q $ non-defect 
arcs in $L$. 
Given an encoding $L$ and vertex $v\in [n]$, let 
$N_L^+(v)$ denote the set of out-neighbours of $v$ (only along
non-defect arcs), that is,
\[ N_L^+(v) = \{ w \in [n] \setminus \{v\}\mid L(v,w) = 1\}.\]
Similarly, 
\[ N_L^-(v) = \{ w \in [n] \setminus \{v\} \mid L(w,v) = 1\}.\]
We also define the in- and out-neighbourhood when neighbours
along defect arcs are included:
\begin{align*}
 \hatN_L^+(v) = \{ w\in [n]\setminus \{v\} \mid L(v,w)\neq 0\}\\
 \hatN_L^-(v) = \{ w\in [n]\setminus \{v\} \mid L(w,v)\neq 0\}.
\end{align*}

Recall that the arc $(v,w)$ has \emph{tail} $v$ and \emph{head} $w$.
Let $\zeta_v^-$ (respectively, $\eta_v^-$) be the number of 2-defect arcs
(respectively, ($-1)$-defect arcs) with $v$ as head,
and define $\zeta_v^+$, $\eta_v^+$ similarly, for defect arcs with
$v$ as tail.  The directed analogues of (\ref{good-edges}) and
(\ref{all-edges}) are
\begin{align}
   |N_L^-(v)| &= d_v^- - 2\zeta_v^- + \eta_v^-,\qquad
   |N_L^+(v)| = d_v^+ - 2\zeta_v^+ + \eta_v^+, \label{good-edges-directed}\\
   |\hatN_L^-(v)| &= d_v^- - \zeta_v^- + 2\eta_v^-,\qquad
   |\hatN_L^+(v)| = d_v^+ - \zeta_v^+ + 2\eta_v^+\ . \label{all-edges-directed}
\end{align}

We can now give the directed analogue of Lemma~\ref{useful-bounds}.

\begin{lemma}
\label{useful-bounds-directed}
Suppose that $L\in C(p,q)$ and let $a_1,b_1$ be distinct vertices with $L(a_1,b_1)\neq 0$.
\begin{itemize}
\item[\emph{(i)}] The number of ways to choose an ordered pair of vertices
$(a_2,b_2)$ such that $L(a_2,b_2)=1$ and $L(a_2,b_1)=0$, with $a_1,b_1,a_2,b_2$ all distinct,
is at least 
\begin{align*} m - 2p + q -
 & \Bigg( \sdmax\Big(\sdmax -\zeta^-_{b_1} + 2\eta^-_{b_1} + 2\Big) + \eta^-_{a_1} + \eta^+_{b_1} 
-  2(\zeta^-_{a_1} + \zeta^+_{b_1})
 \\ & \hspace*{6cm} {} 
+ \sum_{y\in \hatN^-_L(b_1)} (\eta^+_y - 2\zeta^+_y)\Bigg).
\end{align*}
\item[\emph{(ii)}]
Now suppose that $a_1$, $b_1$, $a_2$, $b_2$ are distinct vertices with $L(a_2,b_2)=1$.
Define
\[ \eta^\ast = \eta^-_{a_1} + \eta^+_{b_1} + \eta^-_{a_2} + \eta^+_{b_2},\qquad
   \zeta^\ast = \zeta^-_{a_1} + \zeta^+_{b_1} + \zeta^-_{a_2} + \zeta^+_{b_2}.
\]
The number of ways to choose an ordered pair of vertices $(a_3,b_3)$ such that
$L(a_3,b_3)=1$ and $L(a_1,b_3)=L(a_3,b_2)=0$, with $a_1,b_1,a_2,b_2,a_3,b_3$ all distinct, 
is at least
\begin{align*}
m - 2p+q - & \Bigg(\sdmax\Big(2\sdmax - (\zeta^+_{a_1} +\zeta^-_{b_2}) + 2(\eta^+_{a_1} + \eta^-_{b_2}) + 4\Big)
  + \eta^\ast - 2\zeta^\ast 
 \\ & \hspace*{2cm} {} 
   + \sum_{x\in\hatN_L^+(a_1)} (\eta^-_x - 2\zeta^-_x) + \sum_{y\in\hatN^-_L(b_2)} (\eta^+_y - 2\zeta^+_y) \Bigg).
\end{align*}
\end{itemize}
\end{lemma}

\begin{proof}
The proof of (i) follows exactly as in the undirected case. For (ii), an upper bound on the number
of bad choices of $(a_3,b_3)$ can be obtained by summing two terms. The first term is the number of choices of
$(a_3,b_3)$ such that $\{ a_3,b_3\}\cap \{a_2,b_2\}\neq \emptyset$ or $L(a_3,b_2)\neq 0$.
An upper bound on this number is given by (i) after replacing $a_1$ by $a_2$ and $b_1$ by $b_2$.
The second term is the number of choices of $(a_3,b_3)$ such that $\{ a_3,b_3\}\cap \{a_1,b_1\}\neq \emptyset$
or $L(a_1,b_3)\neq 0$.  An upper bound on this number is given by (i) after reversing all arcs: that is,
by exchanging the roles of $a$ and $b$, by exchanging the superscripts ``$+$'' and ``$-$'', and (for clarity)
replacing the dummy variable $y$ by $x$.  The proof of (ii) is completed by adding these two terms together
and subtracting them from $m-2p+q$, using the definition of $\eta^\ast$ and $\zeta^\ast$.
\end{proof}

The following lemma is the ``critical lemma'' for the directed case,
which we prove in the irregular setting by adapting the proof of 
Lemma~\ref{3switches}.

\begin{lemma} \label{boundonellezeta}
Suppose that the directed degree sequence $\vecdvec$ satisfies
$\sdmin \geq 1$ and $2\leq \sdmax \leq \frac{1}{4} \sqrt m$. 
Let $Z \in \Omegad$. Then 
$$|\mathcal{L}(Z)| \leq \dfrac{1}{8}  m^8 \, | \Omegad|.$$
\end{lemma}

\begin{proof}
We prove that any $L \in \mathcal{L}(Z)$ can be transformed into 
an element of $\Omegad$ (with no defect arcs) using a sequence of at 
most five directed 3-switches.  
An $A$-switch reduces the number of 2-defect
arcs by one, without changing the number of $(-1)$-defect arcs.
A $B$-switch takes an encoding with no 2-defect arcs, and reduces
the number of $(-1)$-defect arcs by one.
Given an encoding $L\in C(p,q)$, in Phase A we perform an $A$-switch $p$ times,
giving an encoding with no 2-defect arcs.  Then in Phase B we perform a
$B$-switch $q$ times, reaching an element of $\Omegad$ (that is, an encoding
with no defect arcs).

\bigskip

\noindent {\bf Phase A.}
If $p=0$ then Phase A is empty and we proceed to Phase B. 
Now assume that $p\in \{1,2,3\}$ and let $L\in C(p,q)$, where $q\in \{0,1,2,3\}$ and $p+q\leq 5$. 
We want a lower bound on the number of 6-tuples $(a_1, b_1, a_2, b_2, a_3, b_3)$ where a 3-switch can be performed with $L(a_1,b_1) = 2$: performing this switch we will produce $L' \in C(p-1,q)$, as shown in Figure~\ref{Aswitch}.
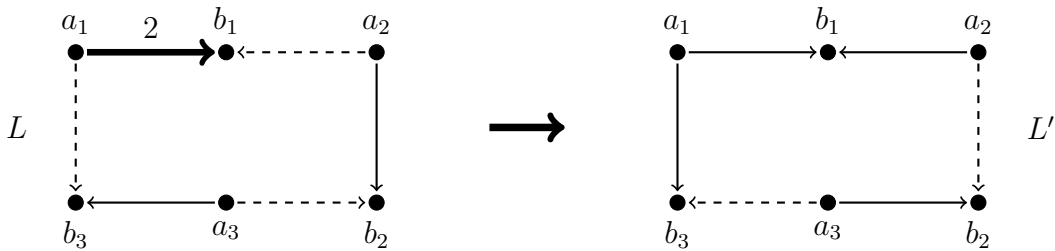
\begin{figure}[ht!] 
\begin{center}
\begin{tikzpicture}[scale = 1]
\draw[fill] (0,0) circle (0.1);
\draw[fill] (0,2) circle (0.1);
\draw[fill] (2,0) circle (0.1);
\draw[fill] (2,2) circle (0.1);
\draw[fill] (4,0) circle (0.1);
\draw[fill] (4,2) circle (0.1);
\draw[fill] (8,0) circle (0.1);
\draw[fill] (8,2) circle (0.1);
\draw[fill] (10,0) circle (0.1);
\draw[fill] (10,2) circle (0.1);
\draw[fill] (12,0) circle (0.1);
\draw[fill] (12,2) circle (0.1);
\draw[line width=2.5pt,->] (0.15,2) -- (1.85,2);
\node [above] at (1,2.05) {$2$};
\draw[thick,->] (4,1.85) -- (4,0.15);
\draw[thick,->] (1.85,0) -- (0.15,0);
\draw[dashed,thick,->] (3.85,2) -- (2.15,2);
\draw[dashed,thick,->] (2.15,0) -- (3.85,0);
\draw[dashed,thick,->] (0,1.85) -- (0,0.15);
\node [below] at (0,-0.1) {$b_3$};
\node [below] at (2,-0.1) {$a_3$};
\node [below] at (4,-0.1) {$b_2$};
\node [above] at (0,2.1) {$a_1$};
\node [left] at (-0.5,1) {$L$};
\node [above] at (2,2.1) {$b_1$};
\node [above] at (4,2.1) {$a_2$};
\draw[thick,->] (8.15,2) -- (9.85,2);
\draw[dashed,thick,->] (12,1.85) -- (12,0.15);
\draw[dashed,thick,->] (9.85,0) -- (8.15,0);
\draw[thick,->] (11.85,2) -- (10.15,2);
\draw[thick,->] (10.15,0) -- (11.85,0);
\draw[thick,->] (8,1.85) -- (8,0.15);
\node [below] at (8,-0.1) {$b_3$};
\node [below] at (10,-0.1) {$a_3$};
\node [below] at (12,-0.1) {$b_2$};
\node [right] at (12.5,1) {$L'$};
\node [above] at (12,2.1) {$a_2$};
\node [above] at (8,2.1) {$a_1$};
\node [above] at (10,2.1) {$b_1$};
\draw [->,line width = 1mm] (5.5,1) -- (6.5,1);
\end{tikzpicture}
\caption{An $A$-switch.}
\label{Aswitch}
\end{center}
\end{figure}
There are $p$ ways to choose the 2-defect arc $(a_1,b_1)$. 
Next, the number of choices of $(a_2,b_2)$ such that $a_1,b_1,a_2,b_2$
are all distinct, $L(a_2,b_2)=1$ and $L(a_2,b_1)=0$ given by
Lemma~\ref{useful-bounds-directed}(ii).
A worst case configuration for the choice of $(a_2,b_2)$ is shown in 
Figure~\ref{fig:PhaseA-worst}.
Here $\eta_{b_1}^-=2$, $\eta_{a_1}^- + \eta_{b_1}^+ = 2$ and
$\sum_{y\in \hatN_L^-(b_1)}\eta_y^+ = 2$.  (Note that we need only
consider configurations of defect arcs which are subdigraphs
of those shown in Figure~\ref{f:configuration}: in particular,
$\eta^-_{v}\leq 2$ for all $v$.)
Clearly $q=3$, and it is possible that $p=2$, though for an upper bound
the additional 2-defect arc will not be incident with any of  $a_1,b_1,a_2,b_2$.
Substituting these values into Lemma~\ref{useful-bounds-directed}(i) shows that 
the number of valid choices for $(a_2,b_2)$ is at least
\begin{align*}
 m - 1 - (\sdmax^2 + 5\sdmax + 2) \geq m - 5\sdmax^2,
\end{align*}
using the fact that $\sdmax\geq 2$.
\begin{figure}[ht!]
\begin{center}
\begin{tikzpicture}
\draw[fill] (4,2) circle (0.1);
\draw[fill] (0,2) circle (0.1);
\draw[fill] (2,2) circle (0.1);
\draw[fill] (2,0) circle (0.1);
\node [above] at (1,2.2) {$2$};
\node [above] at (3.0,2.1) {$-1$};
\node [right] at (2.0,0.7) {$-1$};
\node [below] at (0.9,1.75) {$-1$};
\draw[thick,->] (2,0) -- (2,1.85);
\draw[thick,->] (4,2) -- (2.15,2);
\draw (0,2) [->,thick] arc (120:62:1.9);
\draw (2,2) [->,thick] arc (-60:-118:1.9);
\node [above] at (0,2.1) {$a_1$};
\node [above] at (2,2.1) {$b_1$};
\end{tikzpicture}
\caption{A worst-case configuration for the choice of $(a_2,b_2)$ in Phase A.}
\label{fig:PhaseA-worst}
\end{center}
\end{figure}
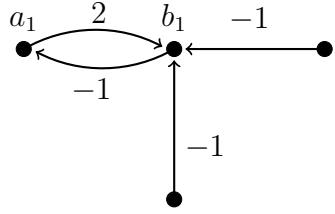

Finally we must choose $(a_3,b_3)$ such that the six chosen vertices
are distinct, $L(a_3,b_3)=1$ and $L(a_3,b_2)=L(a_1,b_3)=0$.
We obtain an upper bound on the number of valid choices for $(a_3,b_3)$ by
applying Lemma~\ref{useful-bounds-directed}(ii). 
A worst case configuration for the choice of $(a_3,b_3)$ is shown below.
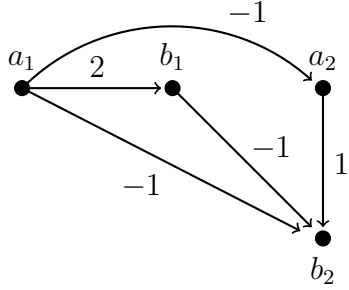
\begin{center}
\begin{tikzpicture}
\draw[fill] (0,2) circle (0.1);
\draw[fill] (2,2) circle (0.1);
\draw[fill] (4,0) circle (0.1);
\draw[fill] (4,2) circle (0.1);
\node [above] at (1,2.0) {$2$};
\node [above] at (3,2.7) {$-1$};
\node [right] at (2.9,1.2) {$-1$};
\node [left] at (2.0,0.7) {$-1$};
\node [right] at (4.0,1.0) {$1$};
\draw[thick,->] (4,2) -- (4,0.15);
\draw[thick,->] (0,2.00) -- (1.85,2.00);
\draw[thick,->] (0,2) -- (3.7,0.10);
\draw[thick,->] (2,2) -- (3.85,0.15);
\draw (0,2) [->,thick] arc (135:48:2.8);
\node [above] at (0,2.1) {$a_1$};
\node [above] at (2,2.1) {$b_1$};
\node [above] at (4,2.1) {$a_2$};
\node [below] at (4,-0.1) {$b_2$};
\end{tikzpicture}
\end{center}
%
\noindent Here 
\begin{align*} 
\eta^* = 2,\quad \zeta^* = 0, \quad & \eta^+_{a_1} + \eta^-_{b_2} = 4,\quad \zeta^+_{a_1} + \zeta^-_{b_2} = 1,\\
\sum_{x\in\hatN_L^+(a_1)} (\eta^-_x - 2\zeta^-_x) = 1, \qquad & \sum_{y\in\hatN^-_L(b_2)} (\eta^+_y - 2\zeta^+_y) = 2.
\end{align*}
Clearly $q=3$, and it is possible that $p=2$, though for a worst case
the remaining 2-defect arc will not be incident with any of the vertices
shown.
Substituting these values into Lemma~\ref{useful-bounds-directed}(ii),
we find that the number of valid choices for $(a_3,b_3)$ is at least
\begin{align*}
  m-1 - (2\sdmax^2 + 11\sdmax + 5) \geq m - 9\sdmax^2.
\end{align*}

Therefore, using the upper bound on $\sdmax$, there are at least
\begin{equation} \label{Adirect}
p\,(m- 5\sdmax^2)(m-9\sdmax^2) \geq \dfrac{77}{256} m^2 > \dfrac{3}{10} m^2
\end{equation}
choices of $A$-switch which can be performed
in $L$ to give an element of $C(p-1,q)$.

Now we consider the reverse operation. Let $L' \in C(p-1,q)$, where
$p \in \{1,2,3 \}$ and  $q \in \{0,1,2,3 \}$. 
We want an upper bound on the number of 6-tuples
$(a_1,b_1,a_2,b_2,a_3,b_3)$ 
such that $L'(a_1,b_1) = L'(a_2,b_1) = L'(a_3,b_2) = L'(a_1,b_3) = 1$ and 
$L'(a_2,b_2)=L'(a_3,b_3)=0$. 
There are at most $m-2(p-1)+q \leq m + 3$ choices for the arc $(a_1,b_1)$ with 
$L'(a_1,b_1)=1$, and then there are at most
\[ (d_{a_1}^+ + (\eta_{a_1}^+)' -1 ) (d_{b_1}^- + (\eta_{b_1}^-)' - 1)
  \leq \sdmax(\sdmax+1) \leq \dfrac{3}{2}\sdmax^2
    \]
choices for the ordered pair $(a_2,b_3)$, since $\sdmax\geq 2$.
Here $(\eta_{b_1}^-)'$, respectively $(\eta_{a_1}^+)'$, denotes the number
of $(-1)$-defect arcs in $L'$ with $b_1$ as head (respectively,
with $a_1$ as tail).
The worst case is when $q=3$ and $\{ (\eta_{a_1}^+)' ,\, (\eta_{b_1}^-)' \} = \{ 1,2\}$.
(Note that $L'(a_1,b_1)=1$, so there can be no ($-1$)-defect
arc from $a_1$ to $b_1$.)

Finally, there are at most $m-2(p-1)+q-3 \leq m$ choices
for the arc $(a_3,b_2)$, since this arc must be distinct from the
3 arcs chosen so far.  Hence, the number of 6-tuples where 
the reverse operation can be performed in $L'$ is at most
\[ \dfrac{3}{2}\sdmax^2\, m(m+3) \leq \dfrac{201}{128}\sdmax^2\, m^2 < \dfrac{8}{5}\, \sdmax^2\, m^2,\]
using the fact that $m\geq 16\sdmax^2 \geq 64$.

Combining this upper bound with $\eqref{Adirect}$ gives, by double counting,
\begin{equation} \label{a}
\frac{|C(p,q)|}{|C(p-1,q)|} \leq \dfrac{16}{3}\sdmax^2.
\end{equation}

\bigskip
\noindent {\bf Phase B.} 
We proceed with Phase B when there are no 2-defect arcs in our encoding. 

Let $L \in C(0,q)$ with $q \in \{1,2,3\}$. We want a lower bound  
to the number of 6-tuples $(a_1,b_1,a_2,b_2,a_3,b_3)$ where a 3-switch 
can be performed in $L$ to produce an encoding $L' \in C(0,q-1)$, 
as shown in Figure~\ref{Bswitch}.

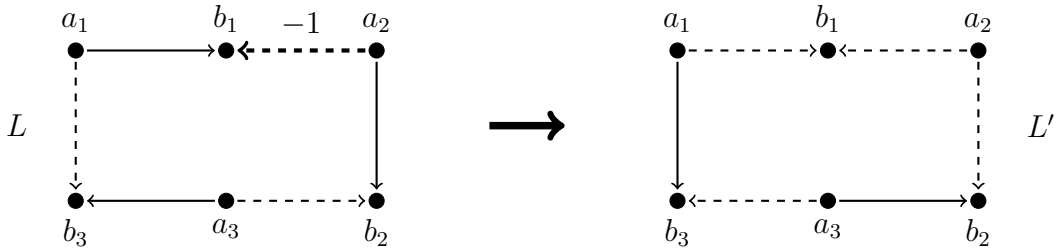
\begin{figure}[ht!] 
\begin{center}
\begin{tikzpicture}[scale = 1]
\draw[fill] (0,0) circle (0.1);
\draw[fill] (0,2) circle (0.1);
\draw[fill] (2,0) circle (0.1);
\draw[fill] (2,2) circle (0.1);
\draw[fill] (4,0) circle (0.1);
\draw[fill] (4,2) circle (0.1);
\draw[fill] (8,0) circle (0.1);
\draw[fill] (8,2) circle (0.1);
\draw[fill] (10,0) circle (0.1);
\draw[fill] (10,2) circle (0.1);
\draw[fill] (12,0) circle (0.1);
\draw[fill] (12,2) circle (0.1);
\draw[thick,->] (0.15,2) -- (1.85,2);
\draw[thick,->] (4,1.85) -- (4,0.15);
\draw[thick,->] (1.85,0) -- (0.15,0);
\draw[ultra thick, dashed,->] (3.85,2) -- (2.15,2);
\draw[dashed,thick,->] (2.15,0) -- (3.85,0);
\draw[dashed,thick,->] (0,1.85) -- (0,0.15);
\node [below] at (0,-0.1) {$b_3$};
\node [below] at (2,-0.1) {$a_3$};
\node [below] at (4,-0.1) {$b_2$};
\node [above] at (0,2.1) {$a_1$};
\node [left] at (-0.5,1) {$L$};
\node [above] at (2,2.1) {$b_1$};
\node [above] at (4,2.1) {$a_2$};
\node [above] at (3,2.05) {$-1$};
\draw[dashed,thick,->] (8.15,2) -- (9.85,2);
\draw[dashed,thick,->] (12,1.85) -- (12,0.15);
\draw[dashed,thick,->] (9.85,0) -- (8.15,0);
\draw[dashed,thick,->] (11.85,2) -- (10.15,2);
\draw[thick,->] (10.15,0) -- (11.85,0);
\draw[thick,->] (8,1.85) -- (8,0.15);
\node [below] at (8,-0.1) {$b_3$};
\node [below] at (10,-0.1) {$a_3$};
\node [below] at (12,-0.1) {$b_2$};
\node [above] at (8,2.1) {$a_1$};
\node [above] at (10,2.1) {$b_1$};
\node [above] at (12,2.1) {$a_2$};
\node [right] at (12.5,1) {$L'$};
\draw [->,line width = 1mm] (5.5,1) -- (6.5,1);
\end{tikzpicture}
\caption{A $B$-switch.}
\label{Bswitch}
\end{center}
\end{figure}

There are exactly $q \geq 1$ choices for the ($-1$)-defect arc $(a_2,b_1)$.
Then we can choose $a_1\in N_L^-(b_1)$ in $d_{b_1}^- + \eta_{b_1}^- \geq 2$
ways, and we can choose $b_2\in N_L^+(a_2)\setminus \{ a_1\}$ in
$d_{a_2}^+ + \eta_{a_2}^+ - 1 \geq 1$ way.  It remains to choose
$(a_3,b_3)$ with all six chosen vertices distinct, $L(a_3,b_3)=1$
and $L(a_3,b_2)=L(a_1,b_3)=0$.  
A lower bound for the number of
valid choices of $(a_3,b_3)$ is given by Lemma~\ref{useful-bounds-directed}(ii).
A worst case configuration for the choice of $(a_3,b_3)$ is shown
in Figure~\ref{worstB}.
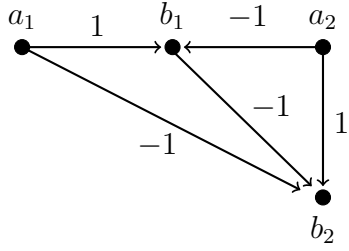
\begin{figure}[ht!]
\begin{center}
\begin{tikzpicture}
\draw[fill] (0,2) circle (0.1);
\draw[fill] (2,2) circle (0.1);
\draw[fill] (4,0) circle (0.1);
\draw[fill] (4,2) circle (0.1);
\node [below] at (1.8,1.0) {$-1$};
\node [above] at (3,2.1) {$-1$};
\node [right] at (2.9,1.2) {$-1$};
\node [right] at (4.0,1.0) {$1$};
\draw[thick,->] (4,2) -- (4,0.15);
\draw[thick,->] (4,2.00) -- (2.15,2.00);
\draw[thick,->] (0,2) -- (3.7,0.10);
\draw[thick,->] (0,2) -- (1.85,2);
\node [above] at (1,2) {$1$};
\draw[thick,->] (2,1.95) -- (3.85,0.15);
\node [above] at (0,2.1) {$a_1$};
\node [above] at (2,2.1) {$b_1$};
\node [above] at (4,2.1) {$a_2$};
\node [below] at (4,-0.1) {$b_2$};
\end{tikzpicture}
\caption{A worst case configuration for the choice of $(a_3,b_3)$ in Phase B.}
\label{worstB}
\end{center}
\end{figure}

\noindent Here $\eta_{a_1}^+ + \eta_{b_2}^- = 3$ and $\eta^\ast = 1$, with
$\sum_{x\in \hatN_L^+(a_1)}\eta_x^- = 
\sum_{y\in\hatN_L^-(b_2)}\eta_y^+ = 3$.
Hence the number of valid choices for $(a_3,b_3)$ is at least
\begin{align*}
m+3 - (2\sdmax^2 + 10\sdmax + 7)\geq m - 8\sdmax^2,
\end{align*}
since $\sdmax\geq 2$.

Putting this together, the number of 6-tuples 
$(a_1,b_1,a_2,b_2,a_3,b_3)$ where a B-switch can be performed in $L$
to give an element of $C(0,q-1)$ is at least
\begin{equation} \label{Bdirect}
2q \, (m-8\sdmax^2) \geq m,
\end{equation}
using the stated upper bound on $\sdmax$.

Next we consider the reverse operation. Let $L' \in C(0,q-1)$, where
$q \in \{1,2,3 \}$. We want an upper bound on the number of 6-tuples with
$L'(a_3,b_2) = L'(a_1, b_3) = 1$ and $L'(a_1,b_1)=L'(a_2,b_1)=L'(a_2,b_2)=
L'(a_3,b_3)=0$.
The encoding $L$ produced by the reverse operation must
be consistent with $Z$ after the reverse operation. 
Since $L(a_2,b_1)=-1$, this implies that $(a_2,b_1)$ must be an arc of $Z$. 
Hence there are at most $m$ choices for $(a_2,b_1)$. 
Next we choose $(a_1,b_3)$ with $L(a_1,b_3)=1$.
For an upper bound, we only check that the arc $(a_1,b_3)$ is not
incident with $b_1$ or $a_2$, ruling out at least 
\[ d_{a_2}^- + d_{a_2}^+ +  d_{b_1}^- + d_{b_1}^+ - 2\geq 2\]
choices.  (The $-2$ term avoids any double-counting.) 
Therefore there are at most
\[ m + (q-1)  - 2 \leq m
\]
choices for $(a_1,b_3)$.
The same argument shows that there are at most $m$ choices for the
arc $(a_3,b_2)$.  It follows that
the number of 6-tuples where the reverse operation can be performed 
in $L'$ to produce an element of $C(0,q)$ is at most $m^3$.

Combining this bound with $\eqref{Bdirect}$ gives, by double counting,
\begin{equation} \label{b}
\frac{|C(0,q)|}{|C(0,q-1)|} \leq m^2.
\end{equation}

\bigskip

\noindent {\bf Consolidation.}
Define
\[ a = \dfrac{16}{3} \sdmax^2,\qquad b = m^2.\]
Then, recalling that $|C(3,3)|=0$,
\begin{align*}
\frac{|\mathcal{L}(Z)|}{|\Omegad|} =  \sum_{p=0}^3\sum_{q=0}^3 \, \frac{|C(p,q)|}{|C(0,0)|} 
  \leq & \,  (1 + a + a^2 + a^3)(1 + b + b^2 + b^3) - a^3b^3\\
\leq & \, \dfrac{1}{8} m^8. 
\end{align*}
The last inequality follows since $\sdmax\geq 2$ amd $m\geq 16\sdmax^2\geq 64$.
This completes the proof of Lemma \ref{boundonellezeta}.
\end{proof}

Now we can prove our upper bound on the mixing time of the directed
switch chain.

\begin{proof}[Proof of Theorem~\ref{main-directed}.]
Recall the definitions from Section~\ref{ss:flow}.
It follows from the bipartite model of directed graphs that
\[ |\Omegad| \leq m! \leq \sqrt{2\pi m}\, \left(\frac{m}{e}\right)^m.\]
Therefore the smallest stationary probability $\pi^*$ satisfies
 $\log( 1/\pi*) = \log|\Omegad| \leq m\log m$. 
Next, observe that $\ell(f)\leq m$ since each transition along a canonical
path replaces an edge of $G$ by an edge of $G'$. Finally, if $e=(Z,Z')$
is a transition of the directed switch chain then
$1/Q(e) = \binom{m}{2}\, |\Omegad|$.  Combining this with
Lemmas~\ref{effee} and~\ref{boundonellezeta} 
gives $\rho(f) \leq \nfrac{1}{4}\, \sdmax^{16} m^{10}$.
Substituting these expressions into (\ref{flowbound}) completes the
proof.
\end{proof}

\end{document}